\documentclass[conference]{IEEEtran}

\usepackage{bm}
\usepackage{enumerate}
\usepackage{threeparttable}
\usepackage{graphicx}
\usepackage{tabularx} 
\usepackage{stfloats}
\usepackage{times}
\usepackage{float}
\usepackage{caption}
\usepackage{mathtools}
\usepackage{subcaption}
\usepackage{amsmath}
\usepackage{wrapfig}
\usepackage{lscape}
\usepackage{pgfgantt}
\usepackage{multirow}
\usepackage{xspace}
\usepackage{calc}
\usepackage{fancyhdr}
\usepackage[shortlabels]{enumitem}
\usepackage{varioref}
\usepackage{tikz}
\usepackage{booktabs}
\usepackage{algorithm}
\usepackage{algcompatible}
\usepackage[noend]{algpseudocode}
\usepackage{soul}
\usepackage{placeins}
\usepackage[normalem]{ulem}
\usepackage{thmtools}
\usepackage{url}
\usepackage{amsfonts}
\usepackage{amsthm}
\usepackage{mathrsfs}
\usepackage{calligra}
\usepackage{pifont}

\ifCLASSINFOpdf
\else
\fi
\newcommand{\sgn}{\ensuremath {\texttt{SGN} {\xspace}}}
\newcommand{\fec}{\ensuremath {\texttt{FEC} {\xspace}}}
\newcommand{\enc}{\ensuremath {\texttt{Enc} {\xspace}}}
\newcommand{\dec}{\ensuremath {\texttt{Dec} {\xspace}}}
\newcommand{\hmac}{\ensuremath {\texttt{HMAC} {\xspace}}}
\newcommand{\mayo}{\ensuremath {\texttt{MAYO} {\xspace}}}

\newcommand{\cert}{\ensuremath {\mathit{cert}}{\xspace}}

\newtheorem{mylemma}{Lemma}{\bfseries}{\rmfamily}
{\bfseries}{\rmfamily}
{\bfseries}{\rmfamily}
{\bfseries}{\rmfamily}
\newtheorem{definition}{Definition}[section]

\newcommand{\emulsion}{\ensuremath {\texttt{EMULSION}}}

\newcommand{\auth}{\ensuremath {\texttt{Authenticate} {\xspace}}}
\newcommand{\broadcast}{\ensuremath {\texttt{Broadcast} {\xspace}}}
\newcommand{\setup}{\ensuremath {\texttt{Setup} {\xspace}}}

\newcommand{\sign}{\ensuremath {\texttt{Sign} {\xspace}}}
\newcommand{\verify}{\ensuremath {\texttt{Verify} {\xspace}}}

\newcommand{\keygen}{\ensuremath {\texttt{KeyGen} {\xspace}}}

\newcommand{\amf}{\ensuremath {\mathit{AMF} {\xspace}}}

\newcommand{\bs}{\ensuremath {\mathit{BS} {\xspace}}}
\newcommand{\ckg}{\ensuremath {\mathit{CKG} {\xspace}}}
\newcommand{\ue}{\ensuremath {\mathit{UE} {\xspace}}}

\newcommand{\id}{\ensuremath {\mathit{ID}}{\xspace}}

\newcommand{\sk}{\ensuremath {\mathit{sk}}{\xspace}}

\newcommand{\pk}{\ensuremath {\mathit{PK}}{\xspace}}

\newcommand{\PC}[1]{
	\vspace{2px}
	\noindent{\bf \IfEndWith{#1}{:}{#1}{#1:}}
}

\newtheorem{theorem}{Theorem}

\begin{document}
%
\title{A Lightweight Post-Quantum Authentication Framework for 5G Base Station Bootstrapping}

\author{ 
\IEEEauthorblockN{Saleh Darzi\IEEEauthorrefmark{1},
                   Mirza Masfiqur Rahman\IEEEauthorrefmark{2}, 
                   Imtiaz Karim\IEEEauthorrefmark{3}, 
                   Rouzbeh Behnia\IEEEauthorrefmark{4}, 
                   Attila A Yavuz\IEEEauthorrefmark{1}, 
                   and Elisa Bertino\IEEEauthorrefmark{2}}
\thanks{\IEEEauthorrefmark{1}Saleh Darzi and Attila A Yavuz are with the Bellini College of AI, Cybersecurity and Computing, University of South Florida; Email: salehdarzi@usf.edu; attilaayavuz@usf.edu}
\thanks{\IEEEauthorrefmark{2}Mirza Masfiqur Rahman and Elisa Bertino are with the Department of Computer Science, Purdue University, Email:rahman75@purdue.edu; bertino@purdue.edu}
\thanks{\IEEEauthorrefmark{3}Imtiaz Karim is with the Department of Computer Science, University of Texas at Dallas, Email:imtiaz.karim@utdallas.edu}
\thanks{\IEEEauthorrefmark{4}Rouzbeh Behnia is with the School of Information Systems at University of South Florida, Email:behnia@usf.edu}
}

\IEEEoverridecommandlockouts
\makeatletter\def\@IEEEpubidpullup{6.5\baselineskip}\makeatother
\IEEEpubid{\parbox{\columnwidth}{
		Network and Distributed System Security (NDSS) Symposium 2027\\
		22--26 March 2027, Seoul, Republic of Korea\\
		ISBN 978-1-970672-09-1\\  
		https://dx.doi.org/10.14722/ndss.2027.[23$|$24]xxxx\\
		www.ndss-symposium.org
}
\hspace{\columnsep}\makebox[\columnwidth]{}}

\maketitle

\begin{abstract}
The absence of authenticated bootstrapping between User Equipments (UEs) and Base Stations (BSs) in 5G leaves System Information Block (SIB) broadcasts unprotected, enabling fake BS attacks, man-in-the-middle interception, and spoofed emergency alerts. Prior efforts such as Public Key Infrastructure (PKI)-based certificate chains, token-based schemes, and identity-based signatures either impose overhead exceeding 5G's strict packet-size constraints or lack post-quantum (PQ) security. Direct NIST-PQC integration is also infeasible; current standards, such as ML-DSA requires 34 fragmented SIB1 packets and up to 5,282\,ms end-to-end delay, and FN-DSA still requires 13 fragments and up to 1,920\,ms. We propose EMULSION, a symmetric chained publicly verifiable authentication framework for 5G/6G BS broadcast authentication. EMULSION is the first framework to exploit native 5G architectural features: fixed SIB transmission windows, millisecond-level time synchronization, and eSIM/USIM credential management to achieve genuine PQ security at symmetric-key efficiency. It uses a timed stream loss-tolerant authentication with HMAC chain anchored by a compact PQ signature~(MAYO) applied once per epoch, fitting authentication within a single packet with no fragmentation and eliminating certificate transmission entirely. Unlike prior schemes, EMULSION extends to the full SIB family (SIB1-SIB21) at no additional per-broadcast cost, demonstrated over the air on SIB1. Evaluated on a real over-the-air 5G testbed, EMULSION achieves 33x lower end-to-end delay and 31x less communication overhead than ML-DSA, and 12x lower delay and 5.4x less overhead than FN-DSA. We formally prove the security of EMULSION and open-source its implementation for public testing and adaptation.
\end{abstract}


%
\IEEEpeerreviewmaketitle

\section{Introduction} \label{sec:introduction} 

Fifth-generation (5G) cellular networks have become the backbone of modern communication infrastructure, with over $3.5$ billion subscribers worldwide and projections exceeding $6$ billion by 2030~\cite{ericsson2025}. Before a User Equipment (UE) can place a call, send a message, or initiate any service, it must first acquire System Information Blocks (SIBs), i.e., broadcast messages that carry essential parameters for cell access, resource allocation, and mobility. SIB1 in particular is the bootstrapping message: it is broadcast every 160~ms on the Downlink Shared Channel (DL-SCH) and must be consumed before any authenticated channel (e.g., RRC, NAS) exists~\cite{3gpp38331}. This unauthenticated bootstrap creates a critical window of vulnerability. A Fake Base Station (FBS) can forge or replay SIB1 to mount a family of downstream attacks, including denial-of-service, location tracking, bidding-down, and man-in-the-middle; all rooted in the absence of broadcast-layer authentication~\cite{lteinspector,park2022doltest,mubasshir2025gotta,sigover}.

Beyond these threats, the long-term security of 5G authentication is fundamentally at risk from quantum-capable adversaries. The classical public-key primitives supporting virtually all deployed cellular security (e.g., Elliptic Curve Cryptography (ECC)) are rendered insecure by Shor's algorithm on a sufficiently large quantum computer~\cite{mitchell2020impact}. Unlike TLS~1.3 or PQ-WireGuard, where large Post-Quantum Cryptography (PQC) keys and certificates can be accommodated in a flexible handshake, SIB messages are constrained broadcast frames with fixed, sub-kilobyte size limits, no interactive exchange, and strict timing windows, making a direct PQC integration far from trivial~\cite{darzi2025future}. Migration guidelines from NIST, ETSI, IEEE, and NSA/CISA jointly emphasize that the PQ transition must begin immediately~\cite{nistpqc2024,etsi2024,nsacisa2025}, yet no existing 5G bootstrapping mechanism satisfies genuine PQ security under the strict protocol constraints.  
Securing the 5G SIB broadcast channel against both classical FBS attacks and future quantum-capable adversaries is thus a pressing open problem.

\noindent\textbf{Prior Research.}
Although prior works have explored SIB authentication and PQ migration for cellular networks, each exhibits at least one of the following limitations:
\textbf{(A)} Existing approaches for SIB1 authentication rely on traditional signature schemes that do not provide PQ security~\cite{ross2024fixing, hussain2019insecure}. 
\textbf{(B)}~Direct substitution of classical signatures with NIST-PQC schemes (ML-DSA, FN-DSA) results in signature and public key sizes that vastly exceed the 372-byte size of a single SIB1 transport block, requiring extensive fragmentation across 12--34 packets and incurring multi-second end-to-end delays~\cite{darzi2025future}.
\textbf{(C)}~PQ-AKA proposals~\cite{damir2022beyond,braeken2025pqaka} target the mutual authentication phase after initial cell access and are orthogonal to the broadcast authentication problem; they do not protect the unauthenticated SIB window.
\textbf{(D)}~Hybrid PQ-KEM approaches~\cite{ko2025aka_hpqc} aim to reduce key exchange overhead but are not designed for the one-to-many broadcast setting of SIBs, where the base station has no per-UE state.

Thus, none of these approaches simultaneously achieve PQ security, single-packet compactness, and operation without any pre-existing authenticated channel---the three properties required for practical 5G broadcast authentication.

\smallskip
\noindent\textbf{Problem.}
With the above limitations in mind, in this paper we aim to answer the following research question: \emph{Is it possible to design a post-quantum secure authentication scheme for 5G SIB broadcasts that doesn't incur heavy fragmentation, demands no pre-existing authenticated channel, and operates within the real-time constraints of the broadcast periodicity?}

\looseness-1 \noindent\textbf{A promising Direction.}
Unlike generic broadcast authentication scenarios, the 5G base station architecture possesses several unique structural features that, when carefully exploited, make efficient PQ broadcast authentication feasible. First, SIBs are transmitted on fixed, periodic windows (e.g., SIB1 every 160~ms), defining natural authentication epochs. Second, modern UEs are equipped with GNSS receivers that provide wall-clock time independently of the cellular base station, and 5G gNBs are themselves GNSS-synchronized for network timing~\cite{3gppnr38104}; this gives both endpoints a shared, sub-microsecond time reference that no FBS adversary can manipulate through forged broadcast messages alone.
Third, the eSIM/USIM ecosystem already supports secure remote provisioning of long-term cryptographic material into tamper-resistant hardware~\cite{3gppesim,gsmasgp}. These three features 
can circumvent the fundamental barriers that have stalled prior approaches. 

\looseness-1 \noindent\textbf{Challenges.}
Realizing this direction requires overcoming three concrete technical challenges.
\emph{(C1)~PQ signature overhead.} Even the most compact PQ signature candidates (e.g., MAYO-2 at 180~B) produce signatures that, when combined with the full certificate chain required for bootstrapping trust, far exceed the 372-byte SIB1 packet size. Naively attaching a PQ signature to each SIB broadcast is therefore infeasible.
\emph{(C2)~Bootstrap without an authenticated channel.} Classical TESLA requires an authenticated delivery of the initial key commitment $K_0$. In the SIB authentication setting, no such channel exists: the UE has not yet completed RRC setup or NAS registration. Any solution must be self-bootstrapping from the first broadcast packet, relying only on material that can be pre-provisioned before the UE encounters the cell.
\emph{(C3)~Real-time and loss-tolerance constraints.} The authentication mechanism must complete within the 160~ms SIB1 periodicity to avoid delaying cell access, must tolerate packet loss inherent in wireless broadcast (especially at cell edges), and must not impose significant computational burden on resource-constrained UEs---all while providing PQ security guarantees.

\noindent\textbf{Our Approach.}
In this paper, we introduce $\emulsion$, a post-quantum broadcast authentication scheme for 5G that exploits the architectural features identified above. $\emulsion$ employs a TESLA-style one-way key chain built on HMAC-SHA256 to authenticate SIBs within each 160~ms broadcast epoch using only symmetric primitives. The asymmetric cost is concentrated in a single epoch-opening packet that carries a compact PQ signature (MAYO-2~\cite{beullensmayo}), anchoring the HMAC chain's root of trust. The long-lived MAYO public key $PK_{\amf}$ is pre-provisioned into the UE's eSIM/USIM, eliminating all certificate chains from the over-the-air path. This design makes the AMF (Access and Mobility Management Function) the sole entity performing asymmetric operations as the core key generator; the BS requires no asymmetric cryptography at runtime.

\looseness-1 \noindent\textbf{Findings.}
We implement $\emulsion$ and a comprehensive set of baselines---NIST-PQC standards (ML-DSA, FN-DSA, SLH-DSA) and classical EC-Schnorr---on a full over-the-air 5G testbed (srsRAN, Open5GS, USRP B210). Our evaluation yields the following key results: (1)~$\emulsion$ achieves an end-to-end authentication delay of \textbf{160.07~ms} versus 1,920~ms for FN-DSA and 5,282~ms for ML-DSA; (2)~it requires only \textbf{324 bytes} of cryptographic overhead, fitting within a single SIB1 transport block with zero fragmentation; (3)~at 10\% packet loss, $\emulsion$'s expected retransmission overhead is \textbf{$\sim$41~bytes} versus $\sim$1,971~bytes for FN-DSA ($\sim$48$\times$ reduction); and (4)~demonstrated on SIB1, $\emulsion$ extends to the complete SIB family (SIB1 through SIB21) at no additional per-broadcast cost.

\subsection{Our Contributions} \label{subsec:contributions}

Our key contributions in this paper are as follows:

\noindent\textbf{Key Observations and Main Idea.} 
Unlike generic broadcast authentication scenarios, the 5G base station authentication protocol possesses unique architectural features that, when exploited, can enable efficient authentication mechanisms:

\begin{enumerate}[leftmargin=*]
\item[(i)] {\em Fixed SIB transmission intervals.} 
SIBs are broadcast on fixed periodic windows, e.g., SIB1 every 160~ms with repeated transmissions as short as 20~ms~\cite{3gpp38331}. 
In 6G, even tighter intervals are anticipated to support sub-millisecond latency targets~\cite{imt2030}. These fixed windows define natural, predictable authentication epochs that a time-aware scheme can exploit without any protocol modification. Importantly, the same periodic broadcast structure applies to all SIB types: SIB2 through SIB9 are also transmitted on DL-SCH within transport blocks of similar size, either periodically broadcast or delivered on-demand via RRC signaling~\cite{3gpp38331}. An authentication mechanism that operates at the transport block level, therefore generalizes across the entire SIB family without per-type adaptation.
 
\item[(ii)] {\em Widely available time synchronization.}
TESLA-style protocols require loose time synchronization between sender and receiver. Modern UEs are equipped with GNSS receivers that provide wall-clock time with sub-microsecond accuracy, independent of the cellular base station's broadcast channel. This independence is critical: the System Frame Number (SFN) carried in the MIB on PBCH is unauthenticated and thus susceptible to manipulation by an FBS adversary, whereas GNSS-derived time originates from an external satellite constellation outside the adversary's control. Since GNSS receivers are standard in smartphones and increasingly prevalent in IoT UEs, and since 5G gNBs are themselves GNSS-synchronized for network timing~\cite{3gppnr38104}, both endpoints share a common, fine-grained time reference that is sufficient to enforce TESLA's safe-packet test without relying on any unauthenticated over-the-air timing signal. 
 
\item[(iii)] {\em eSIM/USIM credential management.} The eSIM infrastructure already supports secure remote provisioning of long-term cryptographic material into tamper-resistant USIMs~\cite{3gppesim,gsmasgp}, the same storage used for Authentication and Key Agreement (AKA) provisioning. This provides a natural anchor for long-lived root key material, removing the need to transmit any certificate or asymmetric public key during SIB authentication. 
\end{enumerate}

By exploiting these three features, we propose $\emulsion$, which harnesses a TESLA-style variant--\emph{Timed Efficient Stream Loss-tolerant Authentication}~\cite{perrig2000efficient}--relying
solely on efficient symmetric primitives (i.e., HMAC) within each fixed SIB transmission window, made viable by 5G's existing precision time synchronization. The HMAC chain is anchored by a compact, PQ digital signature (e.g., MAYO~\cite{beullensmayo}, advanced to Round 3 of NIST's additional signature process) applied once at epoch boundaries rather than per-SIB, providing a root of trust without per-broadcast asymmetric overhead. The eSIM/USIM infrastructure securely provisions this root key, eliminating certificate chains entirely. Table~\ref{tab:ComparisonIntro}
summarizes the performance of $\emulsion$ against
representative baselines; Section~\S\ref{sec:PerformanceEvaluation}
provides evaluation details. The desirable
properties of $\emulsion$ are as follows:\vspace{-1mm}

\begin{table}[h!]
    \centering
    \caption{Comparison of authentication schemes for 5G SIB1 bootstrapping.} 
    \resizebox{0.48\textwidth}{!}{%
    \begin{tabular}{|@{}c@{}||@{}c@{}|@{}c@{}|@{}c@{}|@{}c@{}|@{}c@{}|@{}c@{}|}
         \hline \multirow{2}{*}{\textbf{Scheme}} & \textbf{E2E} & \textbf{Crypto.} & \textbf{Total} & \textbf{PQ} & \textbf{Pkt. Loss} & \textbf{No Frag.} \\
         & \textbf{(ms)} & \textbf{(B)} & \textbf{OTA (B)} & \textbf{Secure} & \textbf{Resilient} & \textbf{Needed} \\ \hline
         \textbf{$\textit{EC\mbox{-}Schnorr}$}~\cite{schnorr1991efficient} & $4.15$ & $256$ & $372$ & \ding{55} & \ding{51} & \ding{51} \\ \hline\hline
         \textbf{$\textit{FN\mbox{-}DSA}$}~\cite{fouque2018falcon} & $1920.67$ & $3792$ & $4836$ & \ding{51} & \ding{55} & \ding{55} \\ \hline
         
         \textbf{$\textit{ML\mbox{-}DSA}$}~\cite{dang2024module} & $5282.47$ & $9884$ & $12648$ & \ding{51} & \ding{55} & \ding{55} \\ \hline\hline
         \textbf{$\emulsion$} ($|\mathcal{C}|{=}2$) & $\mathbf{160.07}$ & $\mathbf{324}$ & $\mathbf{744}$ & \ding{51} & \ding{51} & \ding{51} \\ \hline
    \end{tabular}}
    \begin{minipage}{\linewidth}
    \small \vspace{+1mm}
    \textbf{Note:} E2E delay in milliseconds; overhead in bytes. $\emulsion$ with $PK_{MAYO}$ pre-provisioned in eSIM. PQ counterparts are with 2-level certificates. See Section~\ref{sec:PerformanceEvaluation} for details.
\end{minipage}
    \vspace{-3mm}
    \label{tab:ComparisonIntro}
\end{table}

\begin{itemize}[leftmargin=*]  
\item[-] \looseness-1 \textbf{Symmetric-Optimal Computation and Post-Quantum Efficiency.} Per-SIB authentication in $\emulsion$ is driven by HMAC, delivering PQ security and near-optimal computational efficiency simultaneously. The PQ anchor signature (MAYO) is computed only at epoch boundaries, amortizing its cost across all SIBs within the window. $\emulsion$ is $\mathbf{\sim215\times}$ faster than FN-DSA at the BS and $\mathbf{\sim5\times}$ faster at the UE, with a total end-to-end delay of $\mathbf{160.07}$~ms versus $5{,}282.47$~ms for
ML-DSA. This translates directly to lower network delay and energy savings for resource-constrained UEs.
 
\item[-] \textbf{High Robustness and Security.} (\textit{a})~{\em Packet loss resilience} compact per-SIB tags yield significantly higher packet loss tolerance than fragmented PQC schemes, with expected retransmission overhead of only $\sim41$ bytes at $10$\% loss versus $\sim1971$ bytes for FN-DSA. Beyond tag compactness, TESLA's one-way key chain ensures a UE can authenticate any buffered packet upon receiving its disclosed key, with no retransmission required regardless of intermediate losses.  
(\textit{b})~{\em Error-Correction Compactness}: The redundancy size in forward error correction grows linear with the message size. Since \emulsion~has smaller tag sizes, it achieves the same level of error correction capability with a lesser redundancy compared to asymmetric alternatives (e.g., $\sim$14 versus $\sim$480 redundancy bytes at 10\% loss for $\emulsion$ versus FN-DSA under Reed-Solomon coding).  (\textit{c})~{\em Full SIB coverage}: demonstrated on SIB1, $\emulsion$ extends to the complete SIB family at no additional fragmentation or per-broadcast cost, closing attack surfaces on SIB2–SIB21 that prior schemes leave open.

\item[-] \textbf{Comprehensive Analysis and Open-Source Evaluation.} We implemented $\emulsion$ and also deployed a full set of baselines--NIST-PQC standards~(ML-DSA, FN-DSA, SLH-DSA) and classical schemes~(EC-Schnorr)--on an over-the-air 5G testbed~(srsRAN, Open5GS), covering cryptographic overhead, communication size, end-to-end delay, time synchronization accuracy, and packet loss behavior. 
The source code and testbed implementation are publicly released at: \textcolor{blue}{\url{https://github.com/Prometheused/EMULSION}}
\end{itemize}
\section{Background} \label{sec:preliminaries}

$\lambda$ denotes the security parameter (chain key length, $\lambda = 256$). $t$ denotes the truncated MAC tag length ($t = 128$). $q_v$ denotes the number of verification attempts available to an adversary within a disclosure window.

\noindent \textbf{Notations}: The symbol $||$ denotes concatenation. $x \xleftarrow{\$} S$ denotes uniform random sampling from set $S$. $\lambda$ denotes the security parameter. $F$ and $F'$ denote pseudorandom functions (PRFs) used in TESLA key-chain derivation and MAC-key derivation, respectively. $H$ denotes a collision-resistant hash function (instantiated with SHA-256). $\lfloor \cdot \rfloor$ denotes the floor function. A \emph{time interval} of index $i$ spans the half-open window $[T_0 + i \cdot T_{\mathrm{int}},\; T_0 + (i+1)\cdot T_{\mathrm{int}})$, where $T_0$ is the epoch start time and $T_{\mathrm{int}}$ is the fixed interval duration. $\Delta t$ denotes the maximum clock-synchronization error between the
sender (BS) and the receiver (UE).

\subsection{5G Cellular Network Architecture}\label{subsec:5G}
 
\subsubsection{Network Entities}
A 5G cellular network comprises three principal entities~\cite{fourati2021comprehensive}:
 
\begin{itemize}[leftmargin=*]
    \item \looseness-1 \textit{\textbf{5G Core Network (5GC).}} The 5GC orchestrates service delivery, session management, policy enforcement, and subscriber authentication. Most relevant to this work is the Access and Mobility Management Function (AMF), which terminates Non-Access Stratum (NAS) signaling from UEs and manages registration, connection, and mobility procedures.

 
    \item \looseness-1 \textit{\textbf{User Equipment (UE).}} A UE is any subscriber device at the network edge (e.g., smartphone, IoT sensor, or vehicle), provisioned with a Universal Subscriber Identity Module (USIM) storing a permanent identifier and the cryptographic credentials required for mutual authentication with the core network.

 
    \item \textit{\textbf{Radio Access Network (RAN).}} The RAN consists of base stations (gNBs, denoted BS) and associated UEs, governing radio resource allocation, wireless transmission, and initial access. Critically, the system information messages a BS periodically broadcasts are transmitted in the clear, neither encrypted nor signed, rendering them susceptible to forgery, replay, and manipulation.

\end{itemize}

\subsubsection{Protocol Stack and Initial Access} 
Fig.~\ref{fig:timing} depicts the 5G initial access sequence and protocol layers. At the BS, the topmost control-plane layer is Radio Resource Control (RRC); at the UE and AMF, the NAS layer is stacked above RRC. A master public key $\pk_{\id_0}$ is securely embedded in the USIM and is publicly verifiable; private keys for the AMF and individual BSs are derived from the master secret key $\sk_{\id_0}$ and distributed through authenticated out-of-band channels.  
During initial access, the BS broadcasts System Information (SI), an RRC-layer message announcing cell-specific configuration parameters, comprising the Master Information Block (MIB) and one or more System Information Blocks (SIBs). The MIB is carried on the Physical Broadcast Channel (PBCH) and supplies the scheduling and decoding parameters needed to locate SIB1, which is transmitted on the Downlink Shared Channel (DL-SCH) with a configurable periodicity of 160\,ms and a maximum size of 372\,bytes per 3GPP specifications~\cite{RRCSpec}.


\vspace{-4.5mm}
\begin{figure}[ht!]
	\centering
	\includegraphics[scale=0.42, trim = 5cm 3cm 5cm 0cm, clip, angle=270]{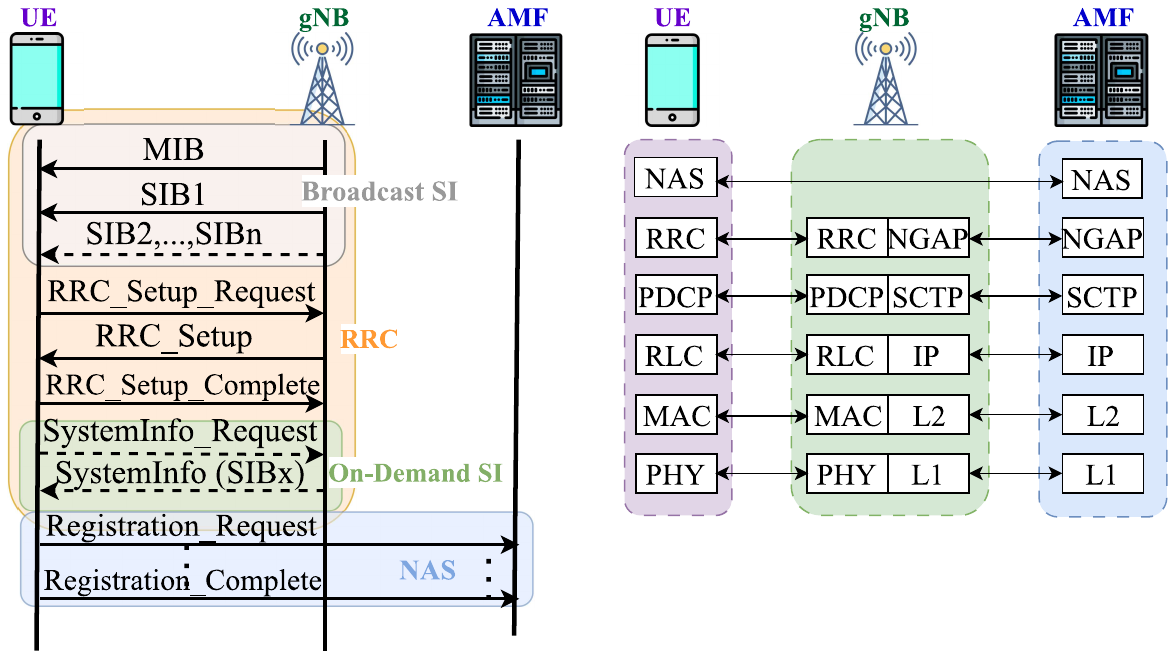}\vspace{-2mm}
	\caption{\small 5G network connection setup and protocol stack.}\vspace{-3mm}
	\label{fig:timing}
\end{figure}

\looseness-1 The UE decodes the MIB from the PBCH to obtain the System Frame Number (SFN), which provides the timing reference needed to locate SIB1 on the DL-SCH. SIB1 carries the core cell configuration including PLMN identity, cell barring status, and scheduling for additional SIBs; mandatory fields include Cell Selection Info (signal quality metrics) and Cell Access Related Info (PLMN identifiers and cell access status), with optional fields such as IMS-Emergency Support included when applicable~\cite{RRCSpec}. Upon decoding SIB1, the UE initiates the RACH procedure and RRC connection setup, after which NAS registration with the AMF occurs. This ordering is security-critical: \emph{SIB1 is consumed before any authenticated channel exists}, since RRC security activation and NAS authentication both occur strictly after the UE has already acted on the received system information. Any SIB authentication mechanism must therefore be self-contained within the broadcast itself, without relying on prior security context.

\noindent\textbf{Additional SIB messages.}
Beyond SIB1, the 3GPP specification defines SIB2 through SIB21, each transmitted over DL-SCH in periodic scheduling windows and serving a distinct purpose. For instance, SIB3 carries intra-frequency neighbor cell lists and reselection parameters, SIB4 provides the inter-frequency equivalents, SIB9 delivers GPS and UTC timing information, and SIB15 conveys disaster-roaming configurations. While the authentication framework developed in this paper targets SIB1, the underlying approach generalizes naturally to other SIB types.


\subsection{Building Blocks}
\label{subsec:cryptoprimitives}

\noindent\textbf{TESLA Broadcast Authentication:}
TESLA (Timed Efficient Stream Loss-tolerant Authentication)~\cite{perrig2000efficient} achieves sender authentication over broadcast channels using only symmetric primitives, at the cost of loose time synchronization. Its core mechanism combines a \emph{one-way key chain} with
\emph{timed key disclosure}: the sender commits to the chain at setup and reveals each key only after a publicly known delay, so that receivers can verify authenticity without shared secrets and without any shared secret or asymmetric operations on every packet.

The chain is built using a one-way function
$F\colon\{0,1\}^\lambda \to \{0,1\}^\lambda$ satisfying \emph{(i) one-wayness}: given $K_i = F(K_{i+1})$, recovering
$K_{i+1}$ is infeasible; and \emph{(ii) target collision resistance (TCR)}: finding $K' \neq K$ with $F(K') = F(K)$ for a given $K$ is infeasible~\cite{perrig2000efficient}.
Notably, $F$ need not be a keyed PRF. The sender samples $K_N \xleftarrow{\$} \{0,1\}^\lambda$ and derives $K_i \leftarrow F(K_{i+1})$ for $i = N{-}1, \ldots, 0$, yielding the \emph{public anchor} $K_0 = F^N(K_N)$. One-wayness ensures a receiver holding $K_i$ cannot recover any $K_j$, $j > i$, while forward iteration trivially yields any
$K_j$, $j < i$. For MAC key derivation, a domain-separated function $F'\colon\{0,1\}^\lambda \to \{0,1\}^\lambda$ produces $K'_i = F'(K_i)$ per interval, preventing structural links between chain and MAC keys. We instantiate both $F$ and the per-packet MAC with \textit{HMAC-SHA-256}~\cite{bellare1996keying}, satisfying one-wayness and TCR under the PRF assumption and aligning with deployed 5G cryptographic stacks~\cite{SecuritySpec}. 

TESLA's security guarantee is that no adversary can forge a MAC tag on any packet that the receiver accepts as authentic, provided the security condition holds (the receiver verifies it received the packet before the corresponding key was disclosed) and $F$ is one-way~\cite{perrig2000efficient}. TESLA tolerates arbitrary packet loss: any subsequent disclosed key $K_j$ allows the receiver to verify all earlier buffered packets from intervals $\leq j$ by forward-iterating $F$, regardless of
intermediate drops. The original TESLA paper~\cite{perrig2000efficient} introduces several variants.


\noindent\textbf{Digital Signature Scheme:}
A digital signature scheme enables a signer to produce an
unforgeable authentication tag on a message that any party holding the public key can verify. 

\begin{definition}
\label{def:sgn}
A Digital Signature Scheme $\sgn$ is a triple of PPT algorithms $(\keygen, \sign, \verify)$: 
\begin{itemize}[leftmargin=*]
    \item[-] $\underline{(\sk, \pk) \leftarrow \sgn.\keygen(1^\lambda)}$: Given security parameter $\lambda$, outputs a secret signing key $\sk$ and a public key $\pk$.
    \item[-] $\underline{\sigma \leftarrow \sgn.\sign(\sk, m)}$: Given $\sk$ and $m \in \{0,1\}^*$, outputs a signature $\sigma$.
    \item[-] $\underline{b \leftarrow \sgn.\verify(\pk, m, \sigma)}$: Returns $1$ (accept) if $\sigma$ is a valid signature on $m$ under $\pk$, and $0$ (reject) otherwise.
\end{itemize}
\end{definition}

Classical signature schemes based on integer factorization or discrete logarithms are broken by Shor's algorithm~\cite{shor1994algorithms}, motivating the adoption of PQ alternatives. NIST concluded its primary PQ standardization process, selecting three signature standards: \textit{ML-DSA} (CRYSTALS-Dilithium)~\cite{dang2024module}, a lattice-based scheme based on Module-LWE and Module-SIS; \textit{SLH-DSA} (SPHINCS+)~\cite{cooper2024stateless}, a stateless hash-based scheme; and \textit{FN-DSA} (Falcon)~\cite{soni2021falcon}, a compact lattice-based scheme over NTRU rings. To encourage diversity in hard problems, NIST also launched an additional signature competition~\cite{nist-addl-sigs} for schemes not based on structured lattices, with first-round selections spanning multivariate, code-based, and symmetric-based candidates.


\noindent\textbf{MAYO as an instantiation of $\sgn$:} Among the third-round candidates of NIST's additional signature
competition~\cite{nist-addl-sigs}, \textit{MAYO}~\cite{beullensmayo} is a leading multivariate scheme based on the Oil-and-Vinegar~(OV) framework, achieving EUF-CMA security under the hardness of the Multivariate Quadratic~(MQ) problem. Its \emph{whipping} technique
yields compact signatures~(186\,B at NIST Level~I) and a manageable public key~(4{,}912\,B), well-suited for bandwidth-constrained 5G broadcasts. MAYO also supports batch verification of $k$ pairs $\{(m_j,\sigma_j)\}_{j=1}^k$ under the same $\pk$ at a cost well below $k$ individual verifications. 

\vspace{-1.5mm}
\section{Approach Overview} 
\label{sec:threatandsecmodel}

\vspace{-1.5mm}
\subsection{Threat Model}
\label{subsec:threatmodel} \vspace{-1mm}
We model the adversary $\mathcal{A}$ as a Quantum Polynomial-Time (QPT) entity with full control over the wireless broadcast channel. Concretely, $\mathcal{A}$ can eavesdrop on all downlink transmissions from any gNB, and may inject, modify, replay, or selectively drop packets at will. $\mathcal{A}$ may also impersonate legitimate BSs
to deceive UEs before any cryptographic protections are established. Under this model, $\mathcal{A}$ mounts three classes of attacks, illustrated in Fig.~\ref{fig:adv_models}:

\begin{itemize}[leftmargin=*]
\item \textbf{Fake Base Station (FBS) Attacks.}
The adversary deploys a rogue base station that spoofs the identity of a legitimate gNB, luring victim UEs into establishing a connection. Once attached, the adversary can launch cascading attacks that exploit vulnerabilities in subsequent protocol stages, including tracking, denial of service, and downgrade attacks~\cite{mubasshir2025gotta}. In \ref{subsec:fbs_characterization}, we further list the attacks that fall under FBS.
\item \textbf{Man-in-the-Middle (MitM) Attacks.}
A MitM adversary positions itself between a UE and a legitimate gNB, intercepting and potentially altering unprotected signaling traffic. This attack is feasible whenever broadcast messages lack cryptographic authentication, such as digital signatures~\cite{rupprecht2019breaking}. 
\item \textbf{Quantum-Capable Adversaries.}
Our threat model further accounts for adversaries that may eventually gain access to quantum computing resources capable of breaking conventional signature schemes~\cite{darzi2023envisioning}.
\end{itemize}

\begin{figure}[t]
\centering
\begin{subfigure}{.32\columnwidth}
  \centering
  \includegraphics[width=0.9\linewidth, trim = 1cm 24cm 10cm 0cm, clip]{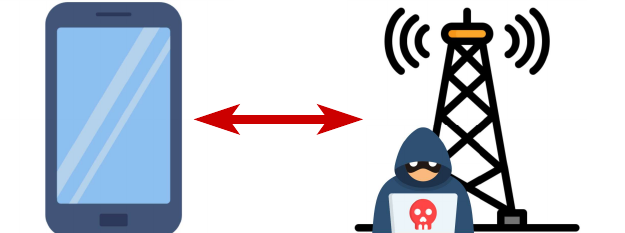}
  \caption{\scriptsize Fake Base Stations}
  \label{fig:sub1}
\end{subfigure}%
\hfill
\begin{subfigure}{.32\columnwidth}
  \centering
  \includegraphics[width=0.9\linewidth, trim = 1cm 24cm 10cm 0cm, clip]{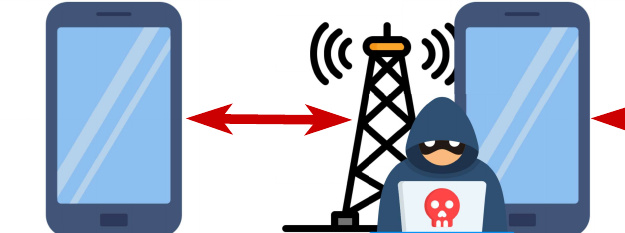}
  \caption{\scriptsize MiTM Attacker}
  \label{fig:sub3}
\end{subfigure}%
\hfill
\begin{subfigure}{.32\columnwidth}
  \centering
  \includegraphics[width=0.9\linewidth, trim = 0cm 24cm 10cm 0cm, clip]{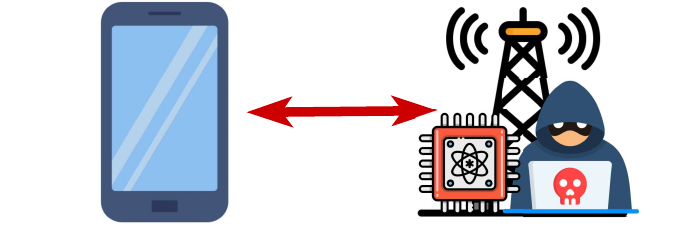}
  \caption{\scriptsize Quantum Adversary}
  \label{fig:sub4}
\end{subfigure}
\caption{Outline of Our Threat Models.}\vspace{-6mm}
\label{fig:adv_models}
\end{figure}

\vspace{-3mm}
\subsection{Security Model} \label{subsec:securitymodel} \vspace{-1mm}
Our security analysis relies on the following definitions that correspond to the security properties of our building blocks.  

\vspace{-1mm}
\begin{definition}[EUF-CMA Security of $\sgn$]
\label{def:sgn-eufcma}
A digital signature scheme $\sgn = (\keygen, \sign, \verify)$ is existentially unforgeable under adaptive chosen-message attack (EUF-CMA) if for all PPT adversaries $\mathcal{A}$, given $\pk$ and adaptive access to a signing oracle $\mathcal{O}_S$, the probability of outputting a valid signature $(m^*, \sigma^*)$ on a fresh $m^*$ not queried to $\mathcal{O}_S$ is negligible: $\mathsf{Adv}^{\mathsf{EUF\text{-}CMA}}_{\sgn,\mathcal{A}}(\lambda) \leq \mathsf{negl}(\lambda)$.
\end{definition}\vspace{-1.5mm}

We instantiate $\sgn$ with MAYO~\cite{beullensmayo}, whose EUF-CMA security reduces to the hardness of the MQ problem, believed to be quantum-resistant. The security proof is generic in $\sgn$ and holds for any EUF-CMA-secure instantiation.

\vspace{-1mm}
\begin{definition}[Security of $F$ and $F'$]
\label{def:chain-security}
Let $F : \{0,1\}^\lambda \rightarrow \{0,1\}^\lambda$ denote the one-way chain function and $F' : \{0,1\}^\lambda \rightarrow \{0,1\}^\lambda$ the MAC-key derivation function, both instantiated with HMAC-SHA-256~\cite{perrig2000efficient}. We require these properties:  
\begin{itemize}[leftmargin=*]
    \item[-] \textit{(i) Pseudorandomness.} $F$ and $F'$ are pseudorandom functions (PRFs): no efficient distinguisher can tell $F(K, \cdot)$ or $F'(K, \cdot)$ apart from a random function with non-negligible advantage, where $K \xleftarrow{\$} \{0,1\}^\lambda$. HMAC-SHA-256 satisfies this under standard assumptions. With $\lambda = 256$~bits chain keys, quantum key recovery via Grover's algorithm requires $2^{128}$ sequential oracle invocations, meeting NIST Level I. 
    \item[-] \textit{(ii) One-Wayness.} Given any chain element $K_i$, where $K_0 = F^\ell(K_\ell)$ and $K_i = F^{\ell-i}(K_\ell)$, no efficient adversary can recover $K_j$ for any $j > i$: $\Pr[\mathcal{A}(K_i) = K_{i+1}] \leq \mathsf{negl}(\lambda)$.
    \item[-] \textit{(iii) Target Collision Resistance (TCR).} Given $K$, no efficient adversary can find $K' \neq K$ such that $F(K') = F(K)$: $\Pr[\mathcal{A}(K) = K' : F(K') = F(K)] \leq
\mathsf{negl}(\lambda)$.
\end{itemize}
\end{definition}


\begin{definition}[$\emulsion$ Security]
\label{def:emulsion-eufcma}
Consider experiment $\mathsf{Exp}^{\mathsf{EUF\text{-}CMA}}_{\emulsion, \mathcal{A}}$: the challenger runs $\emulsion.\setup$ to produce $(\sk_\amf, \pk_\amf)$, chain ${C = \{K_0, \ldots, K_\ell\}}$, and ${\cert_{K_0} \leftarrow \sgn.\sign(\sk_\amf, \mathsf{info})}$ where $\mathsf{info} = K_0 \| \id_\bs \| T_0 \| T_{\mathsf{int}} \| d \| \ell$, giving $\pk_\amf$ to $\mathcal{A}$ (modeling eSIM pre-provisioning). $\mathcal{A}$ may adaptively query a broadcast oracle $\mathcal{O}_B(i, m)$ receiving honest broadcast packets $\Pi_i \leftarrow \emulsion.\broadcast(C, \cert_{K_0}, \mathsf{info}, i, m)$, modeling unrestricted observation of SIB broadcasts. $\mathcal{A}$ wins by outputting $(m^*, i^*, \Pi^*)$ such that $\emulsion.\auth(\pk_\amf, \Pi^*, t_R, \Delta t) = 1$, $(m^*, i^*)$ was never queried to $\mathcal{O}_B$, and $K_{i^*-d}$ has not yet been publicly disclosed. $\emulsion$ is secure if $\mathsf{Adv}^{\mathsf{EUF\text{-}CMA}}_{\emulsion, \mathcal{A}}(\lambda) \leq \mathsf{negl}(\lambda)$ for all QPT $\mathcal{A}$. 
\end{definition}

\subsection{FBS Attack Characterization}
\label{subsec:fbs_characterization}

An FBS is not itself an attack but a \emph{capability} that
enables a set of downstream exploits. Every FBS-driven
compromise follows the same structure: the rogue cell
broadcasts unverifiable SIBs, the \ue{} camps on it, and the
adversary exploits pre-security-context signaling.
Mubasshir~\emph{et al.}~\cite{mubasshir2025gotta} enumerate
$21$ distinct multi-step attacks, all conditioned on this
unauthenticated broadcast, spanning denial of service,
protocol downgrade, identity exposure, device fingerprinting,
battery exhaustion, and public-warning misinformation
(Table~\ref{tab:fbs-taxonomy}). Notably, many attacks depend
on SIB2--SIB8, not SIB1 alone: reselection priorities steer
victims onto rogue cells, and public-warning payloads reside in
SIB6--SIB8. Schemes that protect only SIB1 leave these
vectors intact.

Detection-based defenses~\cite{mubasshir2025gotta} achieve
$96\%$ FBS detection accuracy but are inherently
\emph{post hoc}: the \ue{} must attach before a trace can be
classified. $\emulsion$ solves this preemptively: an unverifiable
transport block is rejected before cell selection, so
the \ue{} never camps on the rogue cell. Parameter cloning
confers no advantage, since it reproduces content but not a
valid authentication tag. The rightmost column of
Table~\ref{tab:fbs-taxonomy} shows the resulting coverage.
Injection-delivered variants (overshadowing, paging hijacking)
are reduced to detected denial of service rather than
prevention; while paging-message authentication remains
orthogonal~\cite{RRCSpec}.

\begin{table}[t]
\centering
\scriptsize
\setlength{\tabcolsep}{2.5pt}
\caption{\small FBS-enabled multi-step attacks~\cite{mubasshir2025gotta}.}
\label{tab:fbs-taxonomy}
\resizebox{\columnwidth}{!}{%
\begin{tabular}{|l|l|l|c|}
\hline
\textbf{Category} & \textbf{Attacks} & \textbf{Dep.} & \textbf{$\emulsion$} \\
\hline\hline
Persistent DoS & Reject, kickoff, counter desync & SIB1 & \ding{51} \\
\hline
Downgrade & RRC redir., cap.\ hijack & SIB2--5 & \ding{51} \\
\hline
Identity exposure & IMSI/SUCI catch, auth relay & SIB1 & \ding{51} \\
\hline
Fingerprinting & MNmap, cap.\ hijack & SIB1 & \ding{51} \\
\hline
Energy drain & Lullaby, energy depletion & SIB1 & \ding{51} \\
\hline
Misinformation & Panic attack & SIB6--8 & \ding{51} \\
\hline
Network abuse & X2/Xn flood, HO hijack & SIB1 & \ding{51} \\
\hline
Injection & Overshadow, paging hijack & PCCH & \ding{51}$^\dagger$ \\
\hline
\end{tabular}}
\begin{minipage}{\linewidth}
    \small \vspace{+1mm}
    \textbf{Note:} \ding{51}~= prevented (forged SIB rejected before cell selection); \ding{51}$^\dagger$~= detected but not prevented. Coverage for SIB2–SIB21 follows from the per-SIB-type chain construction (\ref{subsec:mainoperations}); measurements in section~\ref{sec:PerformanceEvaluation} are for SIB1.
\end{minipage} \vspace{-5mm}
\end{table}

\subsection{Research Challenges and Our Solutions} \label{subsec:challenges}
We arrived at $\emulsion$ by following a chain of design constraints, where each resolution exposes the next obstacle. We articulate these challenges and the design decisions they force, which together motivate the $\emulsion$ framework. The technical realization of each is detailed in Section~\ref{sec:proposedscheme}.

\begin{itemize}[leftmargin=*]
    \item \textbf{C1: Asymmetric PQ signatures cannot fit the broadcast constraint.} SIB broadcasts are fixed, sub-kilobyte, non-interactive frames transmitted on tight periodic windows (e.g., SIB1 is $372$\,B every $20$--$160$\,ms). Since the 5G security stack must adopt NIST-standardized primitives for long-term deployment, any certificate-based scheme that attaches a per-message NIST-PQC signature or certificate chain exceeds the frame budget, forcing fragmentation across many packets ($13$--$34$ packets for NIST-PQC standards) and end-to-end latencies up to $5{,}282$\,ms, which is incompatible with 5G bootstrapping. As our performance evaluation confirms (Section~\ref{sec:PerformanceEvaluation}), this overhead is inherent to per-message asymmetric transmission and cannot be tuned away.

    \item[-] \textbf{Solution.} We remove the asymmetric cost from the per-message path. Each broadcast is authenticated with a symmetric primitive (HMAC), while a single NIST-selected round three PQ signature (MAYO~\cite{beullensmayo}) is computed once per epoch and amortized across all broadcasts within that epoch. This reduces per-SIB overhead to a single compact tag that fits within one packet with no fragmentation, while retaining NIST-grade PQ security at the root of trust. 

    \item \textbf{C2: Symmetric authentication normally presumes a shared secret unavailable at bootstrapping.} Adopting a symmetric approach to satisfy the constraint above requires the sender and receiver to share a key. However, a UE authenticating a BS before any security context is established shares no secret with it, and the broadcast is inherently one-to-many. Conventional symmetric authentication therefore cannot offer public verifiability, and symmetric efficiency and public verifiability appear mutually exclusive.

    \item[-] \textbf{Solution.} We employ delayed-disclosure one-way key chaining (a TESLA-style variant~\cite{perrig2000efficient}), which achieves sender authentication over broadcast channels using only symmetric primitives and no shared secret, by committing to a one-way chain and revealing keys after a bounded delay. This yields the central property that most symmetric schemes lack: public verifiability at symmetric-key efficiency. $\emulsion$ thus attaining the efficiency of a symmetric construction without restriction to pre-shared trust, a gain we substantiate through the remaining challenges.



    \item \textbf{C3: Delayed key disclosure requires a trusted timing and anchor reference at initial access.} The TESLA security condition requires the receiver to bound the sender's transmission time and to trust the chain anchor. At initial access, the UE holds neither an authenticated clock nor a pre-shared anchor, and the timing signals it observes (e.g., SFN on the PBCH) are supplied by the still-unverified BS. A design that assumes these away is vulnerable to replay, where an adversary rebroadcasts a previously valid epoch, its disclosed keys, and its SIB messages to a fresh UE.
    
    \item[-] \looseness-1 \textbf{Solution.} We exploit two trust anchors independent of the unverified base station. For timing, the UE's GNSS receiver provides wall-clock time with sub-microsecond accuracy from a satellite constellation outside the FBS adversary's control, supplying the loose synchronization the TESLA safe-packet test requires~\cite{songtfs2026}. For the chain root, the eSIM/USIM infrastructure~\cite{3gppesim,gsmasgp}, which uses the same tamper-resistant storage as AKA credential provisioning, provisions the root public key offline, eliminating over-the-air certificate transmission. Because both references originate outside the radio channel, epoch freshness is bound to state that the attacker cannot forge, closing the replay surface (Section~\ref{subsubsec:timesync}).


\end{itemize}

\subsection{Scope} \label{subsec:scope} 
$\emulsion$ targets SIB broadcast authentication during the initial bootstrapping phase. It is orthogonal to and does not replace 5G-AKA, which provides UE–5GC mutual authentication and session key establishment. Our model does not address side-channel attacks, physical key extraction, relay attacks, jamming, overshadowing, or passive eavesdropping, each of which requires independent, orthogonal defenses. 


Our scope excludes authentication and key agreement procedures (e.g., 5G-AKA~\cite{rossi2024enhancing, damir2022beyond}) and their post-quantum variants~\cite{pqaka2025tifs, pqaka2022wisec}. These protocols operate at the NAS layer and provide mutual UE--5GC authentication together with session key establishment, i.e., functions that are orthogonal to broadcast bootstrapping authentication. $\emulsion$ targets the pre-NAS broadcast plane: it ensures that the UE receives a genuine SIB1 from a verified BS before any AKA exchange is initiated, thereby preventing FBS driven bootstrapping from propagating into the key agreement phase. Concretely, an operator can enforce this separation through a policy gate that permits AKA initiation only upon receipt of a fresh, $\emulsion$-verified SIB1; this linkage is complementary and preserves the cryptographic structure of AKA.

Overall, we do not address vulnerabilities that lie outside the broadcast authentication plane. In particular, the following are out of scope: (i) \emph{physical-layer attacks} such as, radio-frequency jamming, signal overshadowing, and downlink eavesdropping, which require orthogonal defenses such as anti-jamming techniques or physical-layer encryption; (ii) \emph{hardware attacks} such as physical key extraction from BS hardware; (iii) \emph{UE-to-BS privacy}, such as IMSI/SUPI catching, which is mitigated by NAS-layer SUCI concealment mechanisms~\cite{RRCSpec}; and (iv) \emph{denial-of-service attacks} at the MAC or RLC sublayers (e.g., resource exhaustion via crafted RACH preambles), which target availability rather than authenticity, (v) GNSS spoofing is explicitly excluded from our threat model, consistent with prior work~\cite{singla2021look,songtfs2026}; co-located GNSS spoofing substantially raises the adversary's cost and complexity beyond FBS deployment alone.

While these attacks lie outside our scope, their effect against $\emulsion$ is bounded to denial rather than deception, since none can cause a UE to accept a forged SIB: forgery requires a valid per-packet tag, and thus a chain key, or a valid anchor signature, and thus $sk_{AMF}$---neither of which is available to the adversary. Under overshadowing, a forged SIB lacks a valid tag and fails verification, so the UE aborts before RACH and the attack yields a detected denial of service rather than a fraudulent attachment. Jamming and physical-layer interference prevent delivery entirely, causing bootstrapping to fail as with any broadcast scheme. Relay of unmodified genuine SIBs and downgrade via omission
during incremental deployment are discussed in Appendix~\ref{app:relaydowngrade}. Across all these cases, $\emulsion$ reduces the adversary's capability from substituting system information to merely denying it.
\vspace{-1mm}
\section{The Proposed $\emulsion$ Framework}
\label{sec:proposedscheme}
\vspace{-1mm}
\subsection{Design Rationale and Comparative Justification}
\label{subsec:design}\vspace{-1mm}
\textbf{Infeasibility of Direct NIST-PQC Integration for 5G Authentication.} Direct application of NIST-PQC signatures to SIB broadcasts is fundamentally precluded by 5G's packet-size constraints. Even the most compact standard, FN-DSA (Falcon-512), requires $\approx$6 SIB1 fragments and $\approx$800 ms end-to-end delay in optimized configuration; ML-DSA and SLH-DSA are substantially worse~\cite{darzi2025future}. 



\looseness-1 \textbf{Leveraging 5G Architectural Features.} $\emulsion$ exploits three structural properties of
5G that prior schemes have not exploited. \textit{(i) Fixed SIB transmission windows:} SIB1
is broadcast every 20--160\,ms~\cite{3gpp38331},
defining natural authentication epochs that a
time-aware scheme can exploit without protocol
modification. \textit{(ii) Time synchronization:} Modern UEs maintain wall-clock time via GNSS with sub-microsecond accuracy, independent of any base-station broadcast. Since 5G gNBs are themselves GNSS-synchronized~\cite{3gppnr38104}, both endpoints share a timing reference outside the FBS adversary's control, making the TESLA security condition ($\Delta t \ll T_{\mathrm{int}}$) trivially satisfiable without relying on the unauthenticated SFN. \textit{(iii) eSIM credential provisioning:} The
eSIM infrastructure stores long-term cryptographic
material in tamper-resistant USIMs~\cite{3gppesim,
gsmasgp} at subscription time, allowing the AMF's
MAYO public key (4,912\,B, stored once offline) to
serve as the root of trust with no over-the-air
transmission, eliminating certificate chains
entirely. Together, $\emulsion$ combines a
TESLA-based HMAC chain anchored per epoch
(exploiting (i) and (ii)) with a MAYO-certified
root key provisioned via eSIM (exploiting (iii)),
amortizing the single asymmetric PQ cost across all
SIBs in the epoch and reducing per-broadcast
overhead to a single HMAC computation in
microseconds with no fragmentation.

\textbf{Symmetric Broadcast Authentication via
One-Way Key Chaining.} TESLA imposes minimal runtime overhead: one $F'$ evaluation and one HMAC per interval at the BS, and one HMAC verification per packet at the UE, both completing in microseconds with no fragmentation and no additional SIB transmissions. This contrasts sharply with asymmetric alternatives: ML-DSA requires $\approx$34 SIB fragments and up to 5,282\,ms end-to-end delay; FN-DSA, even in optimized hybrid configuration, requires $\approx$6 fragments and $\approx$800\,ms
(see~\S\ref{sec:PerformanceEvaluation}). We adopt
the packet-loss-tolerant variant of TESLA with
$d$-interval delayed key disclosure, which is the
natural choice given SIB1's fixed broadcast
periodicity~\cite{3gpp38331}. HMAC-SHA-256
instantiates both $F$ and the per-packet MAC,
satisfying the one-wayness and TCR properties
required by TESLA~\cite{perrig2000efficient}, and
the TESLA security condition is natively satisfiable
since time sync is already native at the RAN
level, requiring no additional
timing infrastructure.


\looseness-1 \textbf{PQ Root Certification and Anchor Selection.} MAYO-2~\cite{beullensmayo} offers the most favorable size-security trade-off among NIST PQC candidates and standards for our deployment model: its 186\,B signature fits within a single SIB1 packet (372\,B limit) with headroom for chain parameters, while its 4,912\,B public key is stored once in the UE eSIM and never transmitted over the air. Among NIST additional signature candidates~\cite{nist-addl-sigs}, MAYO-2 is the only multivariate scheme achieving NIST Level~I with a sub-200\,B signature; ML-DSA produces 2,420\,B and FN-DSA 666\,B at the same level, both requiring fragmentation at the root. MAYO's Oil-and-Vinegar structure provides EUF-CMA security under MQ hardness~\cite{beullensmayo}, and
its batch verification property allows the UE to verify multiple per-epoch chain anchor certificates simultaneously, directly benefiting our per-SIB-type chaining design.

\vspace{-1.5mm}
\subsection{$\emulsion$ Framework Initialization}
\label{subsec:setup}\vspace{-1.5mm}

\subsubsection{\textbf{Entities and Trust Model}} 
\label{subsubsec:entities}
$\emulsion$ involves three entities: the AMF (which hosts the Key Management Function, $\ckg$, as a logical sub-function), the BS (gNB), and the UE. The AMF is the root of trust: it holds the MAYO secret key $\sk_\amf$ and is responsible for generating and certifying all cryptographic material distributed to BSs. The UE trusts only the AMF's MAYO public key $\pk_\amf$, provisioned into its eSIM once, offline, via the GSMA Remote SIM Provisioning (RSP) protocol~\cite{gsmaremotesim}. No direct trust relationship between the UE and any BS is assumed; all BS-level trust is derived from the AMF's certification. 

\subsubsection{\textbf{Root Key Generation}} 
\label{subsubsec:rootkeygen}
At system initialization, the $\ckg$ runs $(\sk_\amf, \pk_\amf) \leftarrow \mayo.\keygen(1^\lambda)$. The secret key $\sk_\amf$ is stored securely within the AMF and
never transmitted. The public key $\pk_\amf$ is provisioned to all UE eSIMs via GSMA RSP and remains valid for the lifetime of the deployment (or until a scheduled key rotation). All BS-level authentication traces back to this single root key.

\looseness-1 \textbf{Compromise Resilience and Key Lifecycle.} $\emulsion$ localizes the effect of compromise. A BS holds only the per-epoch chain material provisioned to it over N2, not any long-term secret, so a compromised BS can forge SIBs only until the epoch expires, after which a fresh certified anchor supersedes it. This contrasts with PKI- and IBS-based schemes, where a leaked BS key remains usable until explicit revocation. Only the AMF root key requires lifecycle management: rotation is performed offline over the GSMA RSP channel already used for provisioning, with UEs holding the current and next public key across a transition window, and root revocation is likewise an RSP update, avoiding per-BS revocation machinery. The remaining open case is cross-operator roaming: since trust is anchored in the home operator's pre-provisioned root, a visited network is authenticated only if its root is provisioned in advance or cross-certified by the home operator; absent this, roaming falls back to unauthenticated SIBs as in current deployments.

\subsubsection{\textbf{Time Synchronization}} 
\label{subsubsec:timesync}
$\emulsion$ uses TESLA, whose security condition requires the UE to bound the sender's current time with a maximum error $\Delta t$. Rather than relying on the SFN carried in the unauthenticated MIB, $\emulsion$ derives its timing reference from the UE's GNSS receiver, which provides wall-clock time with sub-microsecond accuracy from a satellite constellation outside the FBS adversary's control. Since 5G gNBs are themselves GNSS-synchronized for network timing~\cite{3gppnr38104}, both endpoints share a common time reference that is independent of any over-the-air broadcast, making the TESLA security condition trivially satisfiable for any reasonable $T_{\mathrm{int}}$. We empirically validate this in \S\ref{sec:PerformanceEvaluation}.

\vspace{-1mm}
\subsection{$\emulsion$ Framework Main Operations}
\label{subsec:mainoperations} \vspace{-1.5mm}
$\emulsion$ is a symmetric chained, publicly verifiable authentication operating in three phases: Setup, Broadcast, and Verify. It is formalized in Algorithms~\ref{alg:emulsion-setup}--\ref{alg:emulsion-verify} with the full protocol flow shown in Fig.~\ref{fig:emulsion}. The algorithms are stated generically for a message $m$; in practice, $m$ is any SIB type (SIB1 through SIB21), and the $(s)$ superscript denoting SIB type is omitted for clarity. All chain variables, intervals, and certificates are implicitly indexed by the active SIB type. \vspace{-1mm}

\begin{algorithm}[ht!]
\small
\caption{$\emulsion.\setup$}
\label{alg:emulsion-setup}
\begin{algorithmic}[1]
\Statex \vspace{-1mm}\hspace{-5mm}$\underline{(\mathcal{C},\,\cert_{K_0}) \leftarrow \emulsion.\setup(1^\lambda,\,\ell,\,d,\,T_0,\,T_{\mathrm{int}})}$:
Run by the AMF once per chain epoch per SIB type.
\State $(\sk_\amf, \pk_\amf) \leftarrow \mayo.\keygen(1^\lambda)$
    \Comment{Root keypair; $\pk_\amf$ provisioned to UE eSIM via GSMA RSP}
\State $K_\ell \xleftarrow{\$} \{0,1\}^\lambda$
    \Comment{Random terminal key; kept secret at BS}
\For{$i = \ell - 1$ \textbf{downto} $0$}
    \State $K_i \leftarrow F(K_{i+1})$
        \Comment{One-way chain; $K_0 = F^\ell(K_\ell)$ is the public anchor}
\EndFor
\State $\mathsf{info} \leftarrow K_0 \| \id_\bs \| T_0 \| T_{\mathrm{int}} \| d \| \ell$
\State $\cert_{K_0} \leftarrow \mayo.\sign(\sk_\amf,\;\mathsf{info})$
\State $\mathcal{C} \leftarrow \{K_0, K_1, \ldots, K_\ell\}$ 
\State Provision $(\mathcal{C},\,\cert_{K_0},\,\mathsf{info})$ to BS over N2 interface
\end{algorithmic}
\end{algorithm}\vspace{-1mm}
\setlength{\textfloatsep}{0pt}

In \textbf{Setup} (Alg.~\ref{alg:emulsion-setup}), the AMF generates the MAYO root key pair and provisions $\pk_\amf$ to UE eSIMs via GSMA RSP~\cite{gsmasgp}. It constructs a one-way key chain of length $\ell$ by sampling a random terminal key
$K_\ell \xleftarrow{\$} \{0,1\}^\lambda$ and iterating $K_i \leftarrow F(K_{i+1})$ down to the public anchor $K_0 = F^\ell(K_\ell)$. The anchor and TESLA parameters are bound as $\mathsf{info} = K_0 \| \id_\bs \| T_0 \| T_{\mathrm{int}} \| d \| \ell$ and certified as $\cert_{K_0} \leftarrow \mayo.\sign(\sk_\amf, \mathsf{info})$. The chain $\mathcal{C}$ and certificate are provisioned to the BS over N2; the BS then operates autonomously for the entire epoch.

\looseness-1 \textbf{Per-SIB-type independent chains.} $\emulsion$ instantiates one independent chain per SIB type: SIB1 messages are authenticated by chain $\mathcal{C}^{(1)}$, SIB2 messages by $\mathcal{C}^{(2)}$, and so on through SIB21. Each chain is parameterized with its own terminal key, transmission interval $T_{\mathrm{int}}$ matching the 3GPP-defined periodicity of that SIB type~\cite{3gpp38331}, and disclosure depth $d$. The AMF runs $\emulsion.\setup$ independently for each type, producing a dedicated certificate $\cert_{K_0}$ for each chain anchor, and provisions all 21 chains and their certificates to the BS in bulk over N2 at epoch setup (precomputed offline). The BS holds the full set of chains and operates autonomously for the entire epoch, selecting the appropriate chain for each SIB type at broadcast time. The per-SIB broadcast and verification procedures follow Algorithms~\ref{alg:emulsion-sign}--\ref{alg:emulsion-verify} identically for each type. Coverage extends from SIB1 to all 21 SIB types by adding one MAYO signing operation per type at epoch setup (all offline at the AMF), incurring zero additional over-the-air overhead per broadcast, since each chain's per-SIB cost remains a single HMAC tag regardless of SIB type. MIB coverage is also achieved at zero additional cost by including the 3-byte MIB content in the per-packet HMAC input. Details can be found in Appendix~\ref{app:mib}.

\vspace{-2mm}
\begin{algorithm}[ht!]
\small
\caption{$\emulsion.\broadcast$}
\label{alg:emulsion-sign}
\begin{algorithmic}[1]
\Statex \vspace{-1mm}\hspace{-5mm}$\underline{\Pi_i \leftarrow \emulsion.\broadcast(\mathcal{C},\,\cert_{K_0},\,\mathsf{info},\,i,\,m)}$: Run by the BS at each broadcast interval $i$.
\State $K'_i \leftarrow F'(K_i)$
    \Comment{Derive MAC key via $F'$}
\State $\tau_i \leftarrow \hmac(K'_i,\; m \| i \| \id_\bs)$
\State $K_{\mathrm{disc}} \leftarrow K_{i-d}$ \textbf{if} $i \geq d$,
    \textbf{else} $\bot$
    \Comment{Disclose key from $d$ intervals ago}
\If{$i = 1$}
    \State $\Pi_i \leftarrow m \| i \| \tau_i \| K_{\mathrm{disc}} \|
           \mathsf{info} \| \cert_{K_0}$
        \Comment{Epoch-opening packet: include anchor certificate}
\Else
    \State $\Pi_i \leftarrow m \| i \| \tau_i \| K_{\mathrm{disc}}$
        \Comment{Steady-state packets: HMAC-only}
\EndIf
\State Broadcast $\Pi_i$ over DL-SCH
\end{algorithmic}
\end{algorithm}\setlength{\textfloatsep}{0pt}\vspace{-2mm}

In \textbf{Broadcast} (Alg.~\ref{alg:emulsion-sign}), at interval
$i$ the BS derives MAC key $K'_i \leftarrow F'(K_i)$ and computes $\tau_i \leftarrow \hmac(K'_i, m \| i \| \id_\bs)$. It discloses $K_{i-d}$, enabling the UE to retroactively verify earlier buffered packets. The epoch-opening packet ($i=1$) carries $\mathsf{info}$ and $\cert_{K_0}$; all subsequent packets carry only $m$, $i$, $\tau_i$, and $K_{\mathrm{disc}}$, yielding a payload of $\approx$131\,B that fits within a single SIB packet without fragmentation.

\begin{algorithm}[ht!]
\small
\caption{$\emulsion.\auth$}
\label{alg:emulsion-verify}
\begin{algorithmic}[1]
\Statex \vspace{-1mm}\hspace{-5mm}$\underline{\{1,\perp\} \leftarrow \emulsion.\verify(\pk_\amf,\,\Pi_i,\,t_R,\,\Delta t)}$:
Run by the UE upon receiving $\Pi_i$.
\State Parse $\Pi_i \rightarrow (m,\,i,\,\tau_i,\,K_{\mathrm{disc}},\,
    [\mathsf{info},\,\cert_{K_0}])$
    \Comment{Bracketed fields present only when $i = 1$}
\If{$i = 1$}
    \State Parse $\mathsf{info} \rightarrow
        (K_0,\,\id_\bs,\,T_0,\,T_{\mathrm{int}},\,d,\,\ell)$
    \If{$\mayo.\verify(\pk_\amf,\;\mathsf{info},\;\cert_{K_0}) \neq 1$}
        \State \Return $\perp$
            \Comment{Anchor not certified by AMF}
    \EndIf
    \State Store $K_0$ and TESLA parameters as trusted state
\EndIf
\State $i_{\max} \leftarrow
    \lfloor(t_R + \Delta t - T_0)\,/\,T_{\mathrm{int}}\rfloor$
\If{$i_{\max} \geq i + d$}
    \State \Return $\perp$
        \Comment{Security condition violated: $K_{i-d}$ disclosed}
\EndIf
\State Buffer $(m,\,i,\,\tau_i)$
\If{$K_{\mathrm{disc}} \neq \bot$}
    \State Let $j \leftarrow i - d$
    \If{$F^{j}(K_{\mathrm{disc}}) \neq K_0$}
        \State \Return $\perp$
            \Comment{Disclosed key inconsistent with trusted anchor}
    \EndIf
    \State $K'_j \leftarrow F'(K_{\mathrm{disc}})$
    \State Retrieve buffered $(m_j,\,j,\,\tau_j)$ 
    \If{$\hmac(K'_j,\;m_j \| j \| \id_\bs) = \tau_j$}
        \State \Return $1$ and \textsc{Accept}$(m_j)$
    \EndIf
\EndIf
\State \Return $\perp$
\end{algorithmic}
\end{algorithm}
\setlength{\textfloatsep}{0pt}

In \textbf{Verify} (Alg.~\ref{alg:emulsion-verify}), the UE
performs the following steps. (i)~On the epoch-opening packet ($i=1$), it verifies $\cert_{K_0}$ via $\mayo.\verify$ under $\pk_\amf$ from the eSIM, and stores $K_0$ and the TESLA parameters as trusted state for this chain. (ii)~It checks the TESLA security condition $i_{\max} < i + d$, discarding the packet if violated. (iii)~It buffers $(m, i, \tau_i)$ pending key disclosure. (iv)~Upon receiving $K_{\mathrm{disc}}$, it verifies chain consistency: $F^j(K_{\mathrm{disc}}) \stackrel{?}{=} K_0$, $j = i - d$. (v)~It derives $K'_j \leftarrow F'(K_{\mathrm{disc}})$ and verifies the HMAC tag of the buffered packet. The UE accepts $m_j$ only after all steps pass. The maximum authentication delay is $d \cdot T_{\mathrm{int}}$, configurable per SIB type to match its 3GPP-defined transmission
periodicity~\cite{3gpp38331}.
 
\vspace{-0.7mm}
\textbf{Packet Loss Tolerance and Optional FEC.} $\emulsion$ adopts the packet-loss-tolerant TESLA variant with $d$-interval delayed key
disclosure~\cite{perrig2000efficient}: a UE that misses interval $i$ can still verify its buffered $m_i$ upon receiving $K_{i-d}$ in any later packet, with no retransmission required. Each broadcast
packet is self-contained, so a single received packet suffices to authenticate all buffered messages with disclosed keys. For deployments
requiring additional robustness against burst errors (e.g., dense urban or high-interference environments), $\emulsion$ supports an optional FEC layer applied to the HMAC-only payload before broadcast. Steady-state packets carry ${\approx}$131\,B, leaving ${\approx}$241\,B of headroom within the 372\,B SIB1 limit~\cite{3gpp38331} to accommodate FEC redundancy without fragmentation. The BS applies $\fec.\enc$ before broadcast and the UE applies $\fec.\dec$ before HMAC verification; all other protocol steps are unchanged. FEC and HMAC serve orthogonal roles (e.g., channel reliability and cryptographic authenticity, respectively) and any
standard code fitting within the available headroom may be used, including Reed-Solomon (MDS-optimal), 5G NR polar codes (native UE hardware support~\cite{3gppnr38104}), or Raptor codes (rateless, adaptive to variable loss rates). The FEC layer introduces no new cryptographic assumptions and does not affect the PQ security of the HMAC chain or the MAYO root anchor.

\begin{figure}
    \centering
    \includegraphics[width=\columnwidth]{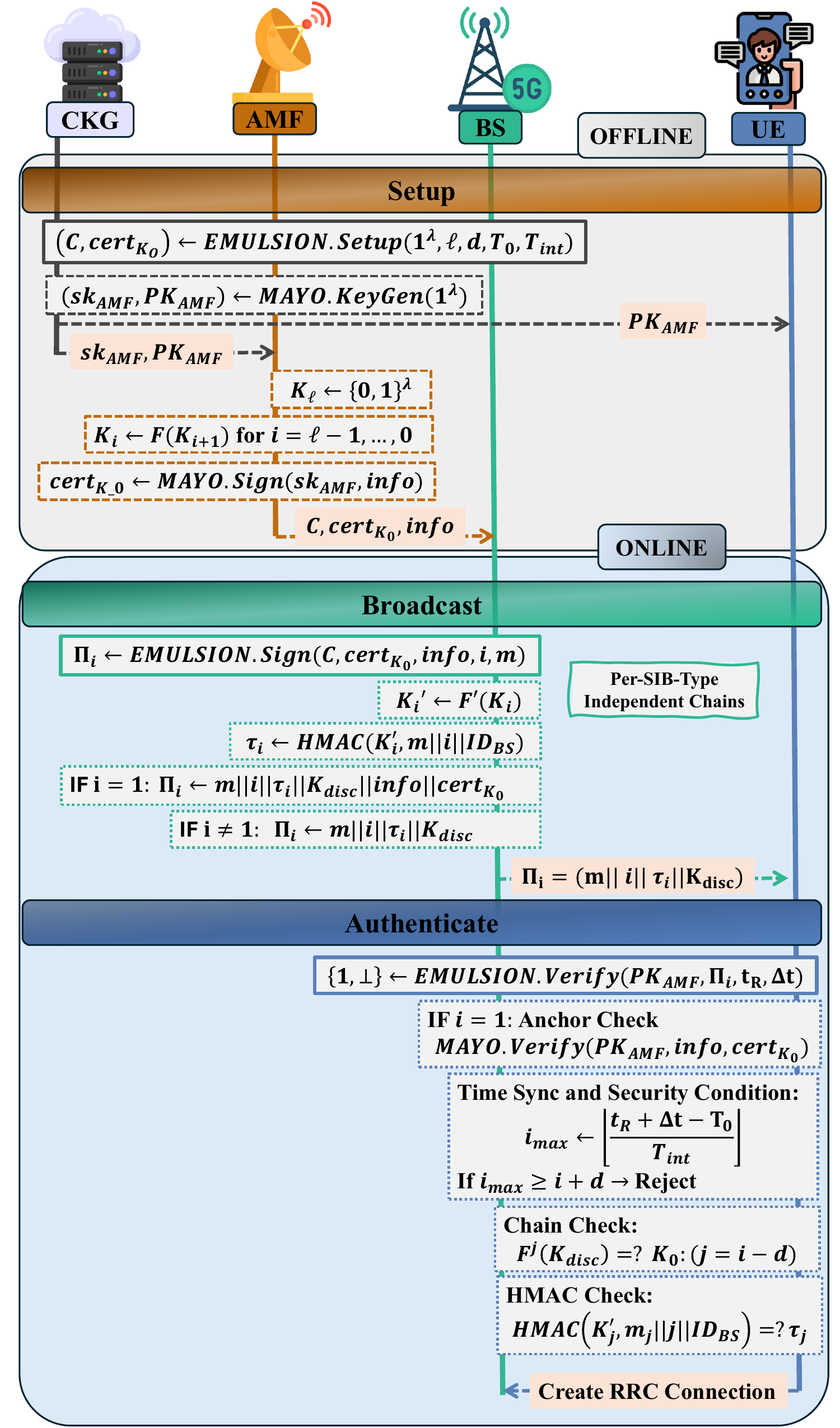}
    \caption{$\emulsion$ protocol flow.}\vspace{-1mm}
    \label{fig:emulsion}
\end{figure}

\vspace{-2mm}
\section{Security Analysis}
\label{sec:securityanalysis} \vspace{-1mm}
We prove $\emulsion$ achieves Definition~\ref{def:emulsion-eufcma} under the security properties of Definitions~\ref{def:sgn-eufcma} and~\ref{def:chain-security}. Any winning adversary $\mathcal{A}$ must succeed at one of three tasks: forge the $\sgn$ anchor certificate $\cert_{K_0}$ produced by $\emulsion.\setup$, forge a per-packet HMAC tag on a SIB broadcast by $\emulsion.\broadcast$ without the chain key, or substitute a fraudulent chain key that passes the consistency check in $\emulsion.\auth$. The following lemmas reduce each task to one of the underlying hard problems with full security proofs presented in Appendix. \vspace{-1mm}

\begin{mylemma}[Anchor Certificate Unforgeability]
\label{lem:anchor}
\textit{If $\mathcal{A}$ forges a fresh $\cert^*_{K_0}$ accepted by $\sgn.\verify$, passing the certificate check in $\emulsion.\auth$, there exists a QPT $\mathcal{B}_1$ breaking EUF-CMA of $\sgn$: $\Pr[\mathcal{A}~\emph{forges}~\cert_{K_0}] \leq \mathsf{Adv}^{\mathsf{EUF\text{-}CMA}}_{\sgn,\mathcal{B}_1}(\lambda)$}. 
\end{mylemma}\vspace{-1.5mm}

\begin{mylemma}[Per-Packet MAC Unforgeability]
\label{lem:mac}
\textit{If $\mathcal{A}$ forges a valid HMAC tag on a fresh ${(m^*, i^*)}$ not broadcast by $\emulsion.\broadcast$, without access to ${K_{i^*-d}}$, there exists a QPT $\mathcal{B}_2$ breaking PRF security of $F'$: ${\Pr[\mathcal{A}~\emph{forges MAC}] \leq \ell \cdot \mathsf{Adv}^{\mathsf{PRF}}_{F', \mathcal{B}_2}(\lambda)}$, where $\ell$ accounts for guessing the target interval $i^*$ uniformly among $\ell$ epochs}.
\end{mylemma}\vspace{-1mm}

\begin{mylemma}[Chain Key Recovery Infeasibility]
\label{lem:chain}
\textit{If $\mathcal{A}$ outputs $K^*_{\mathsf{disc}} \neq K_{i^*-d}$ with $F^j(K^*_{\mathsf{disc}}) = K_0$ where $j = i^*-d$, passing the chain check in $\emulsion.\auth$, there exists a QPT $\mathcal{B}_3$
breaking one-wayness or TCR of $F$: $\Pr[\mathcal{A}~\emph{forges chain key}] \leq \mathsf{Adv}^{\mathsf{OW\text{-}TCR}}_{F, \mathcal{B}_3}(\lambda)$}. 
\end{mylemma}

\begin{theorem}[$\emulsion$ EUF-CMA Security]
\label{thm:emulsion}
\textit{For any QPT adversary $\mathcal{A}$ against $\emulsion$ over a chain of length $\ell$, there exist a QPT $\mathcal{B}_1, \mathcal{B}_2, \mathcal{B}_3$ such that: $\mathsf{Adv}^{\mathsf{EUF\text{-}CMA}}_{\emulsion, \mathcal{A}}(\lambda) \;\leq\; \mathsf{Adv}^{\mathsf{EUF\text{-}CMA}}_{\sgn, \mathcal{B}_1}(\lambda) + \ell \cdot \mathsf{Adv}^{\mathsf{PRF}}_{F',\mathcal{B}_2}(\lambda) +\; \mathsf{Adv}^{\mathsf{OW\text{-}TCR}}_{F, \mathcal{B}_3}(\lambda) + \mathsf{negl}(\lambda)$}.
\end{theorem}

\section{Performance Evaluation}
\label{sec:PerformanceEvaluation}
This section presents a comprehensive evaluation of $\emulsion$ against NIST-PQC standards and conventional authentication schemes for 5G initial bootstrapping. \vspace{-2.5mm}

\subsection{Configuration and Experimental Setup}
\label{subsec:Configuration} 
\noindent\textbf{Hardware:} We assessed the efficiency of $\emulsion$ protocol on a system equipped with a standard desktop equipped with a $12^{th}$ Gen Intel Core ${i7-12700H} @ 3.50~GHz$, $16~GiB$ RAM, a $512~GiB$ SSD, and Ubuntu $22.04.4~LTS$. Real network packets were investigated using the Network Signal Guru Android app installed on a OnePlus Nord 5G smartphone~\cite{netsigguru}.

\noindent \textbf{Over-the-Air Testbed.} We deploy an SDR-based testbed for over-the-air evaluation using the open-source srsRAN and Open5GS stacks. srsUE and srsgNB run on two USRP B210 devices connected to the same host via USB~3.0, with a Leo Bodnar GPSDO providing a stable 10\,MHz reference clock. The srsgNB connects to an Open5GS core and communicates with srsUE over the air. Since commercial UE basebands are closed-source, we adopt a best-effort methodology using this widely used srsRAN and Open5GS platforms, consistent with prior work on 4G/5G bootstrapping security~\cite{ross2024fixing, hussain2019insecure}.


\noindent\textbf{Libraries:} We employed the OpenSSL library\footnote{OpenSSL Library: \url{https://openssl-library.org/}} for cryptographic primitives such as hash functions and elliptic curve operations (e.g., point multiplication, modular arithmetic), the Open Quantum-Safe library\footnote{Open Quantum-Safe Library: \url{https://openquantumsafe.org/}} for NIST-PQC schemes used in the baseline comparisons and for the MAYO-2 signing and verification operations within $\emulsion$, and the blst\footnote{BLST Library: \url{https://github.com/Chia-Network/bls-signatures}} library for the $BLS$ signature.

\noindent\textbf{Parameter Selection:} We configured the post-quantum security to NIST Level~I~\cite{alagic2022status}, which provides quantum resistance approximately equivalent to $128$-bit classical security. For the conventionally secure baselines, all elliptic curve operations were performed over \textit{secp256k1}, defined on a $256$-bit prime field. In $\emulsion$, both the one-way chain PRF~$F$ and the key derivation PRF~$F'$ are instantiated as HMAC-SHA256 with distinct key generation to ensure cryptographic independence. 
MAC tags are truncated to 16 bytes, giving $\approx 2^{128}$ forgery resistance; Grover offers no speedup here, as forging a tag requires online verification attempts against the receiver rather than offline search.


\noindent \textbf{Evaluation Metrics and Rationale:} \label{subsec:EvaluationMetrics}
Quantitative metrics include computational costs (signing, verification, and per-packet MAC operations), 5G processing delay (the time network entities spend handling cryptographic material in packets and transmitting them, excluding the cryptographic computations themselves), cryptographic overhead (signature, key, and certificate sizes), total over-the-air (OTA) communication overhead, and end-to-end (E2E) delay. Qualitative evaluation considers system architecture, PQ security guarantees, and resilience to packet loss.

\noindent \textbf{srsRAN Configuration:} For the first SIB1 message, the srsRAN gNB utilizes $79$ bytes out of the allowed $372$ bytes. All subsequent evaluations report the computational and communication overhead for a $79$-byte SIB1 message, with $293$ bytes available per packet for authentication material, as reflected in the tables. Even considering slightly larger SIB1 configurations observed from real networks, our results remain consistent. Moreover, while we present the evaluation of $\emulsion$ on SIB1, the scheme generalizes to all SIB types, as the one-way chain and MAC-based authentication operate independently of the specific SIB content.

\noindent\textbf{Baseline Selection:}
For PQ baselines, we consider the NIST-standardized $\textit{ML\mbox{-}DSA}$~\cite{dang2024module} (lattice-based) and $\textit{FN\mbox{-}DSA}$~\cite{soni2021falcon} (lattice-based), both in homogeneous
2-level certificate chains (e.g., $\textit{FN\mbox{-}DSA}$-$\textit{FN\mbox{-}DSA}$, $\textit{ML\mbox{-}DSA}$-$\textit{ML\mbox{-}DSA}$) and in a hybrid configuration where the root certificate uses MAYO and the BS uses $\textit{FN\mbox{-}DSA}$ ($\textit{FN\mbox{-}DSA}$-MAYO) where we allow an optimization in favor of $\textit{FN\mbox{-}DSA}$-MAYO by provisioning the MAYO public key ($\pk_\mayo$) in the eSIM, eliminating the need to transmit it over the air, and also a variant augmented with Reed-Solomon erasure coding (FEC) for packet-loss resilience. Hash-based $\textit{SLH\mbox{-}DSA}$~\cite{cooper2024stateless}, with a $7856$-byte signature, is excluded from detailed comparison due to its prohibitively large size and slow execution time (${\sim}11$~ms signing, ${\sim}0.84$~ms verification).  
For conventionally secure baselines, we include $\textit{EC\mbox{-}Schnorr}$~\cite{schnorr1991efficient} and $\textit{BLS}$~\cite{boneh2001short} with
aggregation, both deployed in a 2-level certificate
hierarchy. While these schemes offer favorable
performance, they lack PQ security guarantees and
serve as reference points for understanding the
computational overhead that $\emulsion$ introduces
relative to classical alternatives.
\vspace{-1mm}

\begin{table*}[ht]
    \caption{Comparison of candidate signature schemes for authenticating SIB1.}
    \label{tab:PerformanceComparison}
    \centering
    \begin{tabular}{|c||@{}c@{}|c|c|c|@{}c@{}|c|c|c|c|}
         \hline \multirow{2}{*}{\textbf{Scheme}} & \textbf{System Architecture} & \textbf{Sign} &\textbf{Ver}  & \textbf{5G} & \textbf{Crypto.} & \textbf{Total} & \textbf{E2E} &\textbf{Auth. Success} \\
         &\textbf{and Features}&\textbf{Delay}& \textbf{Delay (ms)}&\textbf{Delay (ms)} & \textbf{Overhead (B)}& \textbf{OTA (B)} & \textbf{Delay (ms)} & \textbf{@10\% Loss} \\\hline
         \textbf{\textit{FN-DSA}}~\cite{fouque2018falcon} & 2-Level Certificate & $0.28~ms$ & $0.15$ & $1920.24$ & $3792$ & $4836$ & $1920.67$ & $25.4\%$\\ \hline
          \textbf{\textit{FN-DSA}}-MAYO & $\pk_\mayo$ in eSIM & $0.28~ms$ & $0.07$ & $800.24$ & $1749$ & $2232$ & $800.59$ & $53.1\%$\\ \hline
          \textbf{\textit{FN-DSA}}-MAYO & $\pk_\mayo$ in eSIM, FEC(n,k) & $0.3~ms$ & $0.07$ & $960.24$ & $2051$ & $2604$ & $960.59$ & $85\%$\\ \hline
          \textbf{$\textit{ML\mbox{-}DSA}$} \cite{dang2024module} & 1-Level Certificate & $0.12~ms$ & $0.09$ & $3201.32$ & $6152$ & $7440$ & $3201.53$ & $10.9\%$ \\ \hline
         \textbf{$\textit{ML\mbox{-}DSA}$} \cite{dang2024module} & 2-Level Certificate & $0.12~ms$ & $0.12$ & $5282.23$ & $9884$ & $12648$ & $5282.47$ & $2.8\%$\\ \hline
        
         \hline\hline
         \textbf{$\emulsion$} & $\pk_\mayo$ in eSIM, $|\mathcal{C}|= 2, d=1$ & $1.3~\mu s$ & $0.03$& $160.04$ & $324$ & $744$ & $160.07$ & $90.0\%$\\
         \hline
    \end{tabular}
    \begin{minipage}{\linewidth}
    \small \vspace{+1mm}
    \textbf{Note:} E2E delay presents the total time for signature generation, 5G delay, and full verification. Auth.\ Success reports epoch establishment probability at 10\% packet loss. The impact of varying chain length $|\mathcal{C}|$ across different authentication ratios is discussed in detail in the experimental results section. 
\end{minipage} \vspace{-6mm}
\end{table*}

\subsection{Experimental Results}
\label{subsec:EvaluationResults} \vspace{-1mm}
TABLE~\ref{tab:PerformanceComparison} presents a comprehensive quantitative and qualitative comparison of candidate schemes for SIB1 authentication in the full 5G hierarchical bootstrapping context, covering signing and verification time, 5G processing and transmission delay, cryptographic and total OTA overhead (bytes), and E2E delay. Results are averaged over 10,000 iterations for standalone measurements and 10 iterations for over-the-air testbed runs (due to the requirement for manual intervention).

\subsubsection{Quantitative Comparison} 
This section analyzes the computational and communication overhead of $\emulsion$ alongside PQ and conventionally secure alternatives.
 
\noindent\underline{\textit{Computational Costs:}}  
A distinguishing feature of $\emulsion$ is BS offloading: the MAYO-2 signature ($\mathit{cert}_{K_0}$) is computed by the AMF offline during epoch setup, so the BS performs only two HMAC-SHA-256 operations per SIB1 packet (one for key derivation and one for MAC computation) at a per-packet signing cost of ${\sim}0.6\,\mu$s,
roughly three orders of magnitude faster than any signature-based alternative (FN-DSA: $0.28$\,ms, EC-Schnorr: $0.30$\,ms, BLS: $0.42$\,ms). Since a gNB may serve hundreds of cells simultaneously, this negligible per-packet cost has significant practical impact on BS infrastructure load. 
On the UE side, steady-state verification (all packets after the epoch-opening $P_1$) requires one HMAC chain check and one MAC verification at ${\sim}0.03$\,ms, with the one-time MAYO-2 certificate verification adding another ${\sim}0.03$\,ms amortized over the epoch. Full PQ certificate chains impose heavier verification: ML-DSA-ML-DSA requires $0.12$\,ms and FN-DSA-FN-DSA requires $0.15$\,ms for two-level verification, while conventionally secure schemes incur higher costs due to pairing (BLS: $3.46$\,ms) or multi-level EC point multiplications (EC-Schnorr: $3.80$\,ms).  
The net result is that $\emulsion$'s computational footprint is dominated entirely by the 5G transmission schedule rather than cryptographic processing: the total cryptographic component of $\emulsion$'s E2E delay (${\sim}0.03$\,ms) is negligible compared to the $160$\,ms SIB1 broadcast periodicity, as confirmed in
Table~\ref{tab:PerformanceComparison}.

\noindent\underline{\textit{Communication Overhead:}} 
$\emulsion$ uses two packet types. The epoch-opening packet ($P_1$) carries the SIB1 message (79\,B), interval index (4\,B), MAC tag (16\,B), chain anchor $K_0$ (32\,B), BS identity and metadata (${\sim}34$\,B), and MAYO-2 certificate (186\,B), totaling ${\sim}351$\,B, within the 372-byte limit. Subsequent packets ($P_2, P_3, \ldots$) carry only the message, interval index, MAC tag, and a disclosed key, totaling 131\,B with just 52\,B of authentication overhead. For a chain of length $|\mathcal{C}|$, total OTA is $(351-79) + (|\mathcal{C}|-1)\times 52$\,B per authentication payload; considering full 372-byte SIB1 blocks, $|\mathcal{C}|{=}2$ requires 744\,B over 2 packets, $|\mathcal{C}|{=}3$ requires 1,116\,B over 3 packets, and $|\mathcal{C}|{=}4$ requires 1,488\,B over 4 packets, with every packet
fitting within a single SIB1 transport block and zero fragmentation required.  
By comparison, full PQ certificate chains impose dramatically higher overhead. The ML-DSA-ML-DSA requires 9,884\,B of cryptographic material across 34 packets (12,648\,B total OTA), and FN-DSA-MAYO with $\pk_\mayo$ pre-provisioned requires 1,749\,B across 6 packets (2,232\,B total OTA), rising to 2,051\,B across 7 packets (2,604\,B total OTA) when Reed-Solomon FEC is added for loss resilience. $\emulsion$ at $|\mathcal{C}|{=}2$ achieves 324\,B cryptographic
overhead ($5.4\times$ lower than FN-DSA-MAYO and $31\times$ lower than ML-DSA-ML-DSA) and even at $|\mathcal{C}|{=}4$ remains at 428\,B, still
$4.1\times$ lower than FN-DSA-MAYO.

\looseness-1 \noindent$\bullet$~\underline{\textit{5G Delay and E2E Latency:}} 
The dominant cost for all schemes is the 5G delay, driven by the SIB1 broadcast periodicity of 20--160\,ms; we evaluate at 160\,ms (the common default)  as the conservative upper bound. For multi-packet schemes, 5G delay is $(\text{packets}-1)\times160$\,ms plus per-packet processing. ML-DSA-ML-DSA requires 34 packets and incurs 5,282.23\,ms (${\sim}5.3$\,s), entirely impractical for bootstrapping before RACH. FN-DSA-FN-DSA requires 13 packets and 1,920.24\,ms; FN-DSA-MAYO with eSIM pre-provisioning reduces this to 6 packets and 800.24\,ms, rising to 7 packets and 960.24\,ms with FEC. $\emulsion$ at $|\mathcal{C}|{=}2$ achieves 160.07\,ms E2E delay, $5\times$ faster than FN-DSA-MAYO (800.59\,ms), $12\times$ faster than FN-DSA-FN-DSA (1,920.67\,ms), and $33\times$ faster than ML-DSA-ML-DSA (5,282.47\,ms). Even at $|\mathcal{C}|{=}4$, $\emulsion$'s 480.11\,ms remains $1.7\times$ faster than the best PQ certificate chain. Table~\ref{tab:PerformanceComparison} reports delay from the epoch-opening packet. A UE arriving mid-epoch waits on average $|C|\cdot T_{int}/2$ more, i.e. $320$~ms at $|C|=2$, still below FN-DSA-MAYO's $800$ms. Larger $|C|$ lowers amortized overhead but raises this wait. The wait applies only at first contact: a UE returning to a cell with cached anchor and parameters verifies from the next steady-state packet.

\begin{figure}[t]
  \centering
  \begin{minipage}[c]{0.52\columnwidth}
    \centering
    \includegraphics[width=\linewidth]{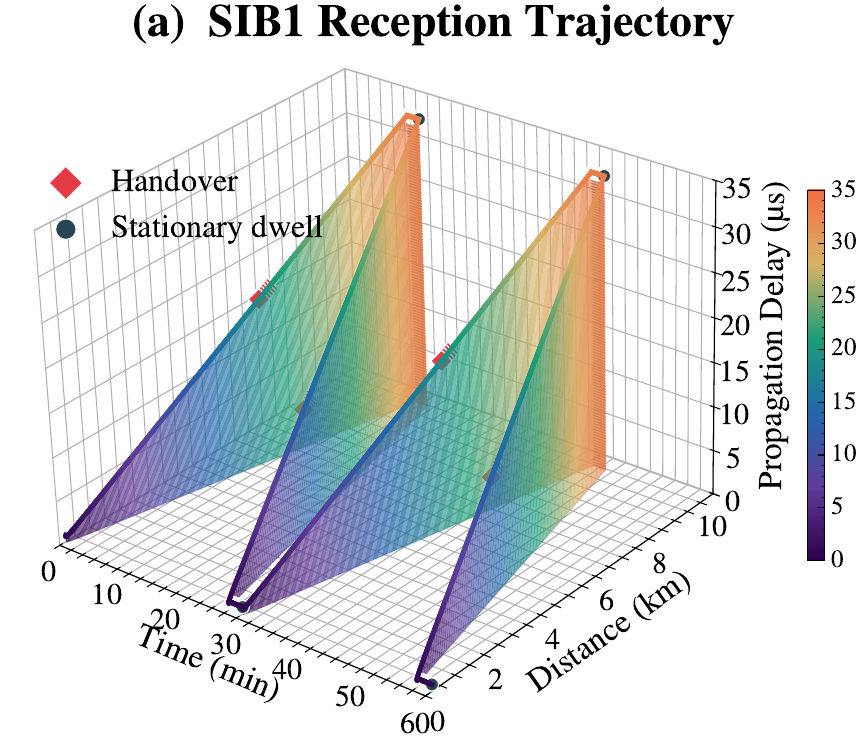}
  \end{minipage}%
  \hfill
  \begin{minipage}[c]{0.46\columnwidth}
    \centering
    \includegraphics[width=0.7\linewidth]{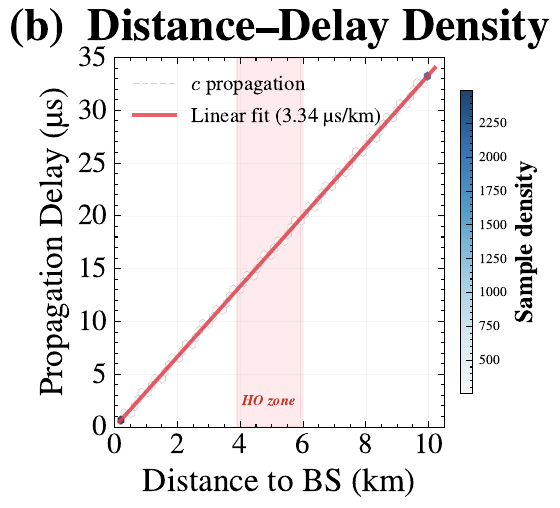}\\
    \includegraphics[width=0.8\linewidth]{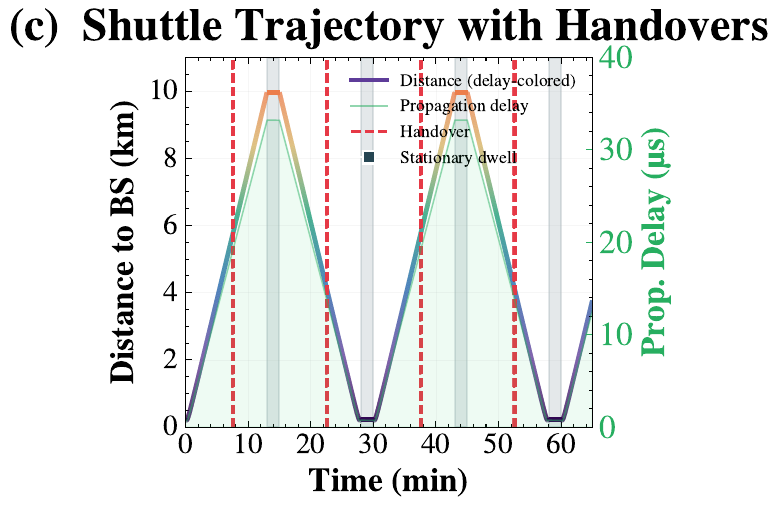}
  \end{minipage}\vspace{-1mm}
  \caption{SIB1 reception characterization on POWDER testbed.}
  \vspace{0.4em}
  \label{fig:shuttle_sib1}\vspace{-2mm}
\end{figure}

To assess whether distance, mobility, and handover introduce measurable overhead, we conduct a real-network experiment on the POWDER testbed~\cite{powder}. Fig.~\ref{fig:shuttle_sib1} shows SIB1 reception traces from a campus shuttle over a ${\sim}60$-minute route with multiple cell attachments and handovers (See Appendix~\ref{app:powder} for experimental setup). Panel~(a) visualizes receptions across time, distance, and propagation delay with visible handover events at cell boundaries. Panel~(b) confirms propagation delay scales linearly at $3.34\,\mu$s/km, reaching at most ${\sim}35\,\mu$s at 10\,km (four orders of magnitude below $\emulsion$'s 160\,ms E2E delay). Panel~(c) traces the UE's distance over the full route, marking handovers and stationary dwells. Across all conditions (with varying distance (0--10\,km), speeds, and repeated handovers) the propagation component remains in the tens of
microseconds, confirming that neither distance, mobility, nor handover introduces any observable impact on authentication delay for $\emulsion$ or any baseline scheme.


\noindent\underline{\textit{UE-Side Overhead:}} $\emulsion$ requires the UE to buffer at most $d$ pending messages awaiting key disclosure
($d\times372$\,B; e.g., 372\,B for $d{=}1$ and 1,116\,B for $d{=}3$). In contrast, PQ certificate-chain schemes require buffering and reassembling fragments across 6--34 consecutive SIB1 packets before any verification can begin, demanding up to 12,648\,B for ML-DSA-ML-DSA and introducing additional sequencing logic into the protocol stack. A single lost fragment forces the UE to discard the entire buffer and restart collection from the next broadcast cycle. In
$\emulsion$, a missed packet affects at most $d$ buffered messages; the next received packet discloses a new key and authentication resumes
immediately without any state resynchronization.


\subsubsection{Qualitative Comparison}
This section examines the structural and security properties that differentiate $\emulsion$ from alternative approaches.

 
\looseness-1 \noindent\underline{\textit{Fragmentation and Protocol Compatibility:}} Every packet of $\emulsion$, including the epoch-opening packet carrying the MAYO-2 certificate, fits within a single 372-byte SIB1 transport block, eliminating fragmentation entirely and preserving full compatibility with the existing 5G protocol stack without RRC modifications, new message types, or application-layer FEC. PQ certificate-chain schemes are fundamentally incompatible with this constraint: ML-DSA-ML-DSA requires ${\sim}10$\,KB of cryptographic material across 34 packets (${\sim}27\times$ the transport block size), and even FN-DSA-FN-DSA requires 13 packets. Fragmentation introduces three practical challenges: \textit{(i)} the UE must buffer and order fragments across multiple broadcast cycles,
complicating the protocol stack; \textit{(ii)} loss of any single fragment invalidates the entire authentication, since partial PQ signatures cannot be verified; and \textit{(iii)} applying FEC to mitigate loss adds further packets and delay, exacerbating rather than resolving the overhead.

 
\looseness-1\noindent\underline{\textit{Post-Quantum Security:}} $\emulsion$ achieves full PQ authentication for SIB broadcasts. The chain anchor $K_0$ is certified by MAYO-2, whose security rests on MQ hardness, while the one-way chain and per-packet HMAC derive security from the PRF assumption via HMAC-SHA-256, providing 128-bit PQ security at NIST Level I. The entire authentication path from AMF root of trust through per-packet MAC verification is therefore PQ-secure. In contrast, BLS (q-SDH), EC-Schnorr (ECDLP), and
Schnorr-HIBS~\cite{singla2021look} all rely on hardness assumptions broken by Shor's algorithm, offering no long-term security against
cryptographically relevant quantum computers.

 
\noindent\underline{\textit{Asymmetric Cost Amortization:}} $\emulsion$ concentrates the asymmetric cost in the epoch-opening packet and amortizes it across the chain. The MAYO-2 signature is computed once per epoch by the AMF offline and verified once by the
UE; all subsequent packets incur only symmetric operations. As $|\mathcal{C}|$ grows, amortized per-packet overhead converges to 52\,B. At $|\mathcal{C}|{=}100$ (${\sim}16$\,s at 160\,ms periodicity), the epoch-opening premium adds only ${\sim}2$\,B per packet on average. This differs fundamentally from approaches such as
EMSS~\cite{perrig2000efficient}, which require a fresh digital signature every $k$-th packet.


\noindent\underline{\textit{Packet-Loss Resilience:}} $\emulsion$ provides inherent loss resilience without FEC. A missed steady-state packet $P_i$ affects at most $d$ messages but causes no cascading failure: the next received packet discloses $K_i$ and authentication resumes immediately. If the epoch-opening packet $P_1$ is missed, the UE waits for the next epoch, which begins with a fresh certified anchor. Certificate-chain schemes have no inherent loss resilience: a single lost fragment out of 6--34 packets causes complete authentication failure, requiring the UE to wait for the next full
broadcast cycle. Reed-Solomon FEC improves resilience but at the cost of additional packets, delay, and decoding complexity. Fig.~\ref{fig:packet_loss_twofig} illustrates these differences. Panel~(a) shows $\emulsion$ (all $|\mathcal{C}|$) maintaining epoch establishment at 90\% even at 10\% packet loss, since only
$P_1$ is critical, while ML-DSA-ML-DSA drops to ${\sim}17\%$ at 5\% loss and FN-DSA-MAYO falls below 50\% at 10\% loss. Panel~(b) shows
steady-state per-message authentication: $\emulsion$ with $|\mathcal{C}|{=}2$, $d{=}1$ authenticates $\sim$80\% of messages at $20\%$ loss, whereas certificate-chain baselines authenticate nothing
until the full chain is successfully reassembled. Moreover, because an \emulsion{} retry costs only one additional epoch ($160$~ms), two consecutive attempts raise the success probability to $1 - p^2 = 99\%$ ($p=10\%$ loss) for a total wait of just $320$~ms. The equivalent two-attempt probability for 
FN-DSA-MAYO ($6$ packets) is only $1-(1-(1{-}p)^6)^2 = 78\%$, 
at a cost of ${\sim}1{,}600$~ms.\vspace{-3mm}

\begin{figure}[!ht]
    \centering
    \includegraphics[width=\columnwidth]{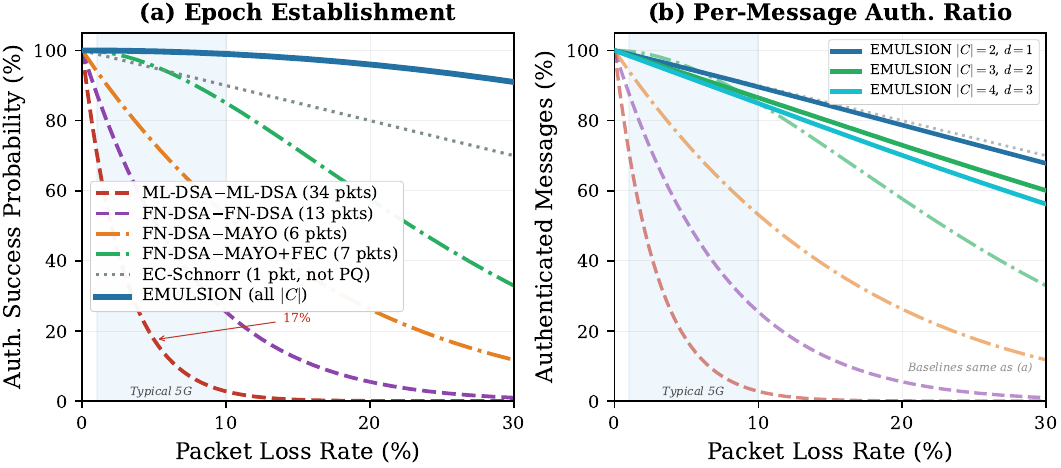}\vspace{-1mm}
    \caption{Packet loss effect on authentication schemes.}\vspace{-3mm}
    \label{fig:packet_loss_twofig}\vspace{-1mm}
\end{figure}

\begin{figure}[!ht]
    \centering
    \includegraphics[width=\columnwidth]{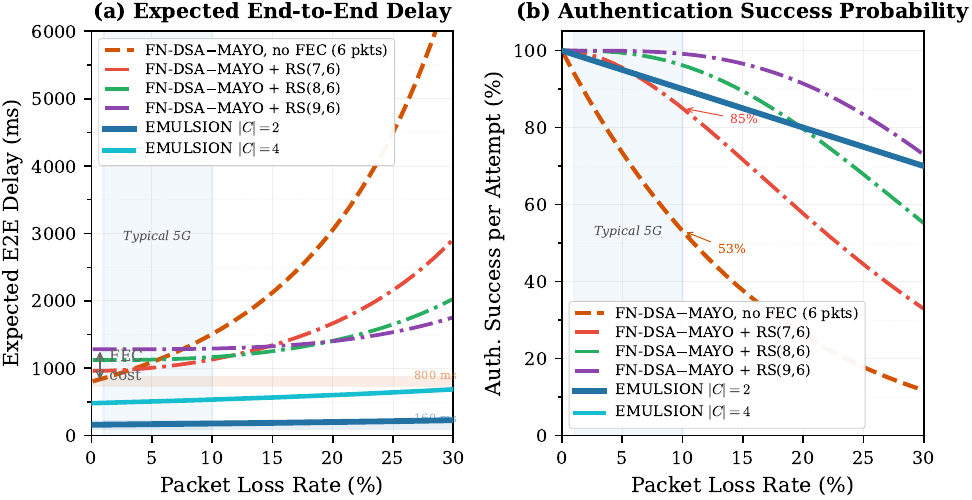}\vspace{-2mm}
    \caption{Tradeoff between end-to-end delay and authentication success based on Reed-Solomon Forward Error Correction (FEC). $\emulsion$ doesn't require FEC and outperforms \textbf{FN-DSA-MAYO} on all FEC settings.}
    \label{fig:fec_tradeoff_twofig}
\end{figure}


Fig.~\ref{fig:fec_tradeoff_twofig} quantifies the delay-resilience tradeoff of Reed-Solomon FEC for FN-DSA-MAYO. Panel~(a) shows that FEC variants RS(7,6) through RS(9,6) reduce expected E2E delay under loss compared to the unprotected 6-packet baseline, but all remain above 800\,ms even at low loss rates, well above $\emulsion$'s 160\,ms at
$|\mathcal{C}|{=}2$. Panel~(b) shows that at 10\% loss, FN-DSA-MAYO without FEC achieves only ${\sim}53\%$ authentication success; while RS(9,6) recovers high success, its E2E delay substantially
exceeds all $\emulsion$ configurations. FEC thus improves FN-DSA-MAYO's loss tolerance but cannot close the gap with $\emulsion$, which requires no error correction overhead.


Overall, $\emulsion$ achieves PQ SIB authentication at 160.07\,ms E2E delay in its minimal configuration ($|\mathcal{C}|{=}2$), which is $5$--$33\times$ faster than PQ certificate-chain alternatives. Its
per-packet BS cost of ${\sim}0.6\,\mu$s (two HMAC operations) is roughly \emph{\textbf{three} orders of magnitude} cheaper than any signature-based scheme, practical for gNBs serving hundreds of cells. Every packet fits withina single SIB transport block, eliminating fragmentation and preserving full 5G protocol stack compatibility. Since $\emulsion$ authenticates over the one-way chain and MAC tags rather than SIB content, it extends naturally to all SIB types without modification, unlike certificate-chain schemes whose fragmentation overhead scales with
the number of protected broadcast channels. Collectively, PQ security, zero fragmentation, HMAC-dominated per-packet cost, inherent loss
resilience, and full SIB generalizability make $\emulsion$ a practical and deployable solution for 5G bootstrapping authentication in the PQ era.


\vspace{-1mm}
\section{Related Work}
\label{sec:RelatedWork} \vspace{-1mm}
We survey prior work on 5G BS authentication and PQ
security for cellular broadcast channels.


\noindent\textbf{PKI-Based BS Authentication.} Early proposals adapted PKI frameworks to authenticate SIB messages. Lee et al.~\cite{lee2009extended} and Zheng~\cite{zheng1996authentication} established certificate-based foundations for mobile network authentication. Hussain et al.~\cite{hussain2019insecure} provided the first systematic attack analysis of 5G bootstrapping and proposed attaching signatures and certificate chains to SIB1/SIB2. Ross et
al.~\cite{ross2024fixing} proposed a ``broadcast-but-verify'' model using a separate \texttt{signingSIB} message to decouple overhead from
SIB1. Gao et al.~\cite{gao2021evaluating} explored delegated signing to reduce per-BS cost, and Wuthier et al.~\cite{wuthier2024fake} combined multi-factor authentication with blockchain-based certificate delivery. 3GPP has explored PKI-based SIB protection in TR~33.809, though SIB1 remains unprotected in the current RRC specification. All PKI-based schemes share a fundamental limitation: certificate chains for AMF
and BS keys routinely exceed the 372-byte SIB1 limit, require fragmentation across multiple packets, and compound verification cost on resource-constrained UEs, while remaining vulnerable to quantum-capable
adversaries.


\looseness-1 \noindent\textbf{Token- and Symmetric-Based Schemes.} BARON~\cite{lotto2023baron} employs symmetric tokens for pre-authentication defense via a Closed Trusted Entity to protect connection initialization and handover in 5G, avoiding asymmetric certificate overhead but leaving SIB contents entirely unprotected: broadcast parameters can be tampered without invalidating any
token. Al-Mekhlafi~\cite{al2024post} extends symmetric security to 5G IoT contexts but similarly does not address SIB broadcast authentication. 
Song et al. ~\cite{songtfs2026} apply TESLA to SIB1 with a Schnorr-like IBS anchor verified once per cell entry, removing per-message signatures but leaving the root of trust on ECDLP and covering SIB1 alone. Substituting a PQ IBS does not help, as none fits a single SIB transport block. 


\noindent\textbf{Identity-Based and Certificate-Free Schemes.} To eliminate certificate chains, Singla et al.~\cite{singla2021look} proposed \textsc{Schnorr-HIBS}, a hierarchical IBS scheme deriving BS and AMF keys from a master key pre-installed in the USIM, achieving
compact overhead within a single SIB1 packet. Ramadan et al.~\cite{ramadan2020identity} explored server-aided IBS to offload UE verification cost. Yu et al.~\cite{yu2024protecting}, Dong et
al.~\cite{dong2025securing}, and Sun and Peng~\cite{sun20255g} further refined two-level HIBSs for LTE/5G. Sengupta and
Lakshminarayanan~\cite{sengupta2024fast} extended this to online-offline threshold IBS for 5G IoT. Darzi et al.~\cite{darzi2025future} presented BORG, a threshold IBS scheme with fail-stop properties that distributes trust across multiple BSs and provides post-mortem forgery detection via a PQ-secured audit log. 
While IBS schemes achieve the best efficiency among conventional approaches, all rely on ECDLP hardness, are broken by quantum-capable adversaries, and are predominantly evaluated for SIB1 only.

\noindent\textbf{Hybrid Post-Quantum Solutions.} A growing body of work combines classical and PQ primitives for 5G/6G security. Vuppala et
al.~\cite{vuppala2023post} and Ko et al.~\cite{ko20255g} proposed hybrid schemes pairing classical key exchange with lattice-based KEMs for
primary authentication. Scalise et al.~\cite{scalise2024applied} analyzed PQ KEMs for 5G/6G core network security, and Attema and de
Kock~\cite{attema2025post} examined PQC deployment challenges in the 5G core. These efforts target unicast session establishment and are structurally incompatible with one-to-many SIB broadcast authentication: KEMs establish shared secrets between two parties, and applying hybrid signatures to SIBs would combine the overhead of two schemes, further
exacerbating fragmentation.


\vspace{-1.5mm}
\section{Conclusion and Future Work} 
\label{sec:conclusion}\vspace{-1.5mm}
\looseness-1 We presented $\emulsion$, a symmetric chained publicly verifiable authentication framework for 5G/6G BS broadcast authentication achieving genuine post-quantum security at symmetric-key efficiency. By exploiting fixed SIB transmission windows, BS-UE time synchronization, and eSIM/USIM credential management, $\emulsion$ harnesses a TESLA-style HMAC chain anchored by a compact MAYO signature applied once per epoch, eliminating certificate chains, avoiding fragmentation, and extending to the full SIB family (MIB and SIB1 through SIB21) at no additional per-broadcast overhead. Evaluated on a real over-the-air 5G testbed, it achieves lower end-to-end delay and communication overhead than ML-DSA with stronger, more durable security guarantees. Future work will extend $\emulsion$ to multi-operator roaming and overshadowing mitigations in the post-quantum setting as 5G transitions toward 6G.




\bibliographystyle{IEEEtran}
\bibliography{References}

\appendices
\section{Security Proofs} \label{sec:securityproof_appendix}

We present the full proofs of Lemmas~\ref{lem:anchor}--\ref{lem:chain} and Theorem~\ref{thm:emulsion}. Each lemma constructs an explicit reduction adversary and analyzes its simulation and success probability.


We use hybrid games $G_0 \to G_1 \to G_2 \to G_3$, where $G_0$ is the real experiment and $\Pr[G_3=1] \leq \mathsf{negl}(\lambda)$. In $G_1$, we abort if $\mathcal{A}$ forges a fresh $\sgn$ certificate passing the check in $\emulsion.\auth$; by Lemma~\ref{lem:anchor}:
$|\Pr[G_0=1] - \Pr[G_1=1]| \leq \mathsf{Adv}^{\mathsf{EUF\text{-}CMA}}_{\sgn, \mathcal{B}_1}$. In $G_2$, we additionally abort if $\mathcal{A}$ forges a valid HMAC tag on a fresh $(m^*, i^*)$ not produced by $\emulsion.\broadcast$; since the TESLA security condition (Definition~\ref{def:chain-security}) prevents access to $K_{i^*-d}$ before its disclosure window expires, by Lemma~\ref{lem:mac}: $|\Pr[G_1=1] - \Pr[G_2=1]| \leq \ell \cdot \mathsf{Adv}^{\mathsf{PRF}}_{F', \mathcal{B}_2}$. In $G_3$, we abort if $\mathcal{A}$ submits a fraudulent $K^*_{\mathsf{disc}}$ passing the chain check in $\emulsion.\auth$; by Lemma~\ref{lem:chain}:
$|\Pr[G_2=1] - \Pr[G_3=1]| \leq \mathsf{Adv}^{\mathsf{OW\text{-}TCR}}_{F, \mathcal{B}_3}$. After $G_3$, $\mathcal{A}$ has no remaining winning strategy: certificate forgeries, MAC forgeries, and chain key forgeries are all excluded, covering every verification path through $\emulsion.\auth$ (Algorithm~\ref{alg:emulsion-verify}), so $\Pr[G_3=1] \leq \mathsf{negl}(\lambda)$. Combining via the triangle inequality
yields the bound in Theorem~\ref{thm:emulsion}.

\noindent\textbf{Lemma~\ref{lem:anchor}} (Anchor Certificate Unforgeability)\textbf{.} \textit{If $\mathcal{A}$ forges a fresh
$\cert^*_{K_0}$ accepted by $\sgn.\verify$, passing the certificate check in $\emulsion.\auth$, there exists a QPT $\mathcal{B}_1$ breaking EUF-CMA of $\sgn$: $\Pr[\mathcal{A}~\text{forges}~\cert_{K_0}] \leq \mathsf{Adv}^{\mathsf{EUF\text{-}CMA}}_{\sgn, \mathcal{B}_1} \lambda)$}. 

\begin{proof}
The reduction exploits the fact that $\mathcal{B}_1$ can simulate the entire $\emulsion$ experiment without knowing $\sk_\amf$, since the chain $\mathcal{C}$ is sampled independently of the signing key. We construct a reduction $\mathcal{B}_1$ that uses $\mathcal{A}$ as a blackbox to break EUF-CMA of $\sgn$ (Definition~\ref{def:sgn}). 

\noindent\textit{\underline{Setup}.} $\mathcal{B}_1$ receives a challenge public key $\pk$ from its EUF-CMA challenger and sets $\pk_\amf := \pk$. It samples the full key chain honestly by picking
$K_\ell \xleftarrow{\$} \{0,1\}^\lambda$ and computing $K_i \leftarrow F(K_{i+1})$ for $i = \ell-1,\ldots,0$, yielding anchor $K_0 = F^\ell(K_\ell)$. It computes $\mathsf{info} = K_0 \| \id_\bs \| T_0 \| T_{\mathsf{int}} \| d \| \ell$ and obtains $\cert_{K_0} \leftarrow \mathcal{O}^{\sgn}_S (\mathsf{info})$ via a single query to its own $\sgn$ signing oracle. $\mathcal{B}_1$ runs $\mathcal{A}$ on input $\pk_\amf$.

\noindent\textit{\underline{Broadcast oracle}.} For each query $\mathcal{O}_B(i, m)$ from $\mathcal{A}$, $\mathcal{B}_1$ computes
$\Pi_i \leftarrow \emulsion.\broadcast(C, \cert_{K_0}, \mathsf{info}, i, m)$ directly using the known chain. No further $\sgn$ signing oracle queries are needed.

\noindent\textit{\underline{Extraction}.} When $\mathcal{A}$ outputs $(m^*, i^*, \Pi^*)$ containing a fresh $\cert^*_{K_0}$ satisfying
$\sgn.\verify(\pk_\amf, \mathsf{info}^*, \cert^*_{K_0}) = 1$ for some $\mathsf{info}^*$ never submitted to $\mathcal{O}_B$, $\mathcal{B}_1$ outputs $(\mathsf{info}^*, \cert^*_{K_0})$ as its $\sgn$ forgery.

\noindent\textit{\underline{Analysis}.} The simulation is perfect: $\mathcal{B}_1$ knows the full chain and answers all broadcast oracle queries honestly. The certificate requires exactly one $\sgn$ signing oracle query, permitted in the EUF-CMA game. Whenever $\mathcal{A}$ forges a fresh certificate, $\mathcal{B}_1$ wins its EUF-CMA game: $\Pr[\mathcal{A}~\text{forges}~\cert_{K_0}] = \Pr[\mathcal{B}_1~\text{wins EUF-CMA}] \leq \mathsf{Adv}^{\mathsf{EUF\text{-}CMA}}_{\sgn, \mathcal{B}_1}(\lambda). \qedhere$
\end{proof}

\noindent\textbf{Lemma~\ref{lem:mac}} (Per-Packet MAC Unforgeability)\textbf{.} \textit{If $\mathcal{A}$ forges a valid HMAC tag on a fresh $(m^*, i^*)$ not broadcast by $\emulsion.\broadcast$, without access to $K_{i^*-d}$, there exists a QPT $\mathcal{B}_2$ breaking PRF security of $F'$: $\Pr[\mathcal{A}~\text{forges MAC}] \leq \ell \cdot \mathsf{Adv}^{\mathsf{PRF}}_{F',\mathcal{B}_2} (\lambda)$, where $\ell$ accounts for guessing the target interval $i^*$ uniformly among the $\ell$ chain epochs}.

\begin{proof}
We construct $\mathcal{B}_2$ that uses $\mathcal{A}$ to break PRF security of $F'$.

\noindent\textit{\underline{Setup}.} $\mathcal{B}_2$ has access to a PRF oracle $\mathcal{O}$ that is either $F'(K^*, \cdot)$ for a hidden random key $K^*$, or a truly random function $R(\cdot)$. It guesses a target interval $i^* \xleftarrow{\$} \{1,\ldots,\ell\}$ uniformly. It samples $K_\ell \xleftarrow{\$} \{0,1\}^\lambda$, computes $K_i \leftarrow F(K_{i+1})$ for all $i$, and obtains $\cert_{K_0}$ via a single $\sgn$ signing oracle query as in Lemma~\ref{lem:anchor}. It runs $\mathcal{A}$ on input $\pk_\amf$.

\noindent\textit{\underline{Broadcast oracle}.} For queries $\mathcal{O}_B(i, m)$ with $i \neq i^*$, $\mathcal{B}_2$ computes
$\tau_i = \hmac(F'(K_i), m \| i \| \id_\bs)$ directly. For $i = i^*$, it uses its PRF oracle: $K'_{i^*} \leftarrow \mathcal{O}(K_{i^*})$, then $\tau_{i^*} = \hmac(K'_{i^*}, m \| i^* \| \id_\bs)$. All other packet fields are computed honestly.

\noindent\textit{\underline{TESLA security condition}.} By Definition~\ref{def:chain-security}, $\mathcal{A}$ does not have access to $K_{i^*-d}$ before the disclosure window expires. Therefore $\mathcal{A}$ cannot compute $K'_{i^*-d} = F'(K_{i^*-d})$ and must forge $\tau^*$ without knowing the MAC key. This is the critical point where the TESLA security condition is invoked in the reduction.

\noindent\textit{\underline{Extraction}.} When $\mathcal{A}$ outputs a valid forgery $\tau^*$ on a fresh $(m^*, i^*)$ at the guessed interval, $\mathcal{B}_2$ distinguishes $\mathcal{O}$ from random: a valid MAC forgery has negligible probability against a truly random function. $\mathcal{B}_2$ outputs ``PRF'' when $\mathcal{A}$ succeeds and ``random'' otherwise.

\noindent\textit{\underline{Analysis}.} The simulation at all $i \neq i^*$ is perfect. At $i^*$, the simulation is indistinguishable from the real game when $\mathcal{O} = F'(K^*, \cdot)$. The guessing step succeeds with probability $1/\ell$: conditioned on $\mathcal{A}$ forging at any interval, it forges at $i^*$ with probability at least $1/\ell$: $\Pr[\mathcal{A}~\text{forges MAC}] \leq \ell \cdot \mathsf{Adv}^{\mathsf{PRF}}_{F',\mathcal{B}_2} (\lambda). \qedhere$
\end{proof}

\noindent\textbf{Lemma~\ref{lem:chain}} (Chain Key Recovery Infeasibility)\textbf{.} \textit{If $\mathcal{A}$ outputs $K^*_{\mathsf{disc}} \neq K_{i^*-d}$ with $F^j(K^*_{\mathsf{disc}}) = K_0$ where $j = i^*-d$, passing the chain check in $\emulsion.\auth$, there exists a QPT $\mathcal{B}_3$ breaking one-wayness or TCR of $F$: $\Pr[\mathcal{A}~\text{forges chain key}] \leq
\mathsf{Adv}^{\mathsf{OW\text{-}TCR}}_{F, \mathcal{B}_3}(\lambda)$}.

\begin{proof}
We construct $\mathcal{B}_3$ that uses $\mathcal{A}$ to break one-wayness or TCR of $F$.

\noindent\textit{\underline{Two cases}.} A fraudulent $K^*_{\mathsf{disc}} \neq K_{i^*-d}$ passing $F^j(K^*_{\mathsf{disc}}) = K_0$ implies one of two things. \textit{Case~1 (Preimage):} $\mathcal{A}$ inverts $F$ along the chain, breaking one-wayness (Definition~\ref{def:chain-security}(ii)). \textit{Case~2 (Collision):} $\mathcal{A}$ finds $K^*_{\mathsf{disc}} \neq K_{i^*-d}$ with $F(K^*_{\mathsf{disc}}) = F(K_{i^*-d})$, breaking TCR (Definition~\ref{def:chain-security}(iii)).

\noindent\textit{\underline{Setup}.} $\mathcal{B}_3$ receives anchor $K_0$ as its OW/TCR challenge. It samples $K_\ell \xleftarrow{\$}
\{0,1\}^\lambda$, computes the full chain $K_i \leftarrow F(K_{i+1})$, and verifies $F^\ell(K_\ell) = K_0$, resampling if necessary. It obtains $\cert_{K_0}$ via a single $\sgn$ signing oracle query and runs $\mathcal{A}$ on input $\pk_\amf$, answering all broadcast oracle queries honestly.

\noindent\textit{\underline{Extraction}.} When $\mathcal{A}$ outputs $K^*_{\mathsf{disc}} \neq K_{i^*-d}$ with $F^j(K^*_{\mathsf{disc}}) = K_0$, $\mathcal{B}_3$ traces the chain from $K^*_{\mathsf{disc}}$ by applying $F$ repeatedly and comparing against the honest chain. The first position $k$ where $F(K^*_k) = F(K_k)$ but $K^*_k \neq K_k$
is a TCR collision (Case~2). If no such position exists, $\mathcal{B}_3$ has found a preimage for $K_0$ under $F^j$ (Case~1). In either case, $\mathcal{B}_3$ outputs the relevant witness.

\noindent\textit{\underline{Analysis}.} The simulation is perfect since $\mathcal{B}_3$ knows the full chain. Every fraudulent $K^*_{\mathsf{disc}}$ passing the chain check in $\emulsion.\auth$ yields a break of one-wayness or TCR: $\Pr[\mathcal{A}~\text{forges chain key}] \leq \mathsf{Adv}^{\mathsf{OW\text{-}TCR}}_{F, \mathcal{B}_3}(\lambda). \qedhere$
\end{proof}

\noindent\textbf{Theorem~\ref{thm:emulsion}} ($\emulsion$ EUF-CMA Security)\textbf{.} \textit{For any QPT adversary $\mathcal{A}$ against $\emulsion$ over a chain of length $\ell$, there exist QPT $\mathcal{B}_1, \mathcal{B}_2, \mathcal{B}_3$ such that: $\mathsf{Adv}^{\mathsf{EUF\text{-}CMA}}_{\emulsion, \mathcal{A}}(\lambda) \leq \mathsf{Adv}^{\mathsf{EUF\text{-}CMA}}_{\sgn,
\mathcal{B}_1}(\lambda) + \ell \cdot \mathsf{Adv}^{\mathsf{PRF}}_{F',\mathcal{B}_2}(\lambda) + \mathsf{Adv}^{\mathsf{OW\text{-}TCR}}_{F, \mathcal{B}_3}(\lambda)
+ \mathsf{negl}(\lambda)$}. 

\begin{proof}
We define hybrid games $G_0, G_1, G_2, G_3$, where $G_0$ is the real
$\mathsf{Exp}^{\mathsf{EUF\text{-}CMA}}_{\emulsion, \mathcal{A}}$ experiment. We bound the gap between consecutive games and show
$\Pr[G_3=1] \leq \mathsf{negl}(\lambda)$.

\noindent\textit{\underline{$G_0$: Real experiment}.} $\mathcal{A}$ interacts with the honest $\emulsion$ challenger. $\Pr[G_0=1] =
\mathsf{Adv}^{\mathsf{EUF\text{-}CMA}}_{\emulsion, \mathcal{A}}(\lambda)$.

\noindent\textit{\underline{$G_0 \to G_1$: Abort on certificate forgery}.} $G_1$ is identical to $G_0$ except the challenger aborts if $\Pi^*$ contains a fresh $\cert^*_{K_0}$ passing $\sgn.\verify$. Any execution where $G_0$ outputs $1$ but $G_1$ outputs $0$ yields a valid $\sgn$ forgery. By Lemma~\ref{lem:anchor}: $|\Pr[G_0=1] - Pr[G_1=1]| \leq \mathsf{Adv}^{\mathsf{EUF\text{-}CMA}}_{\sgn, \mathcal{B}_1}(\lambda)$. 

\noindent\textit{\underline{$G_1 \to G_2$: Abort on MAC forgery}.}
$G_2$ additionally aborts if $\mathcal{A}$'s forgery contains a valid HMAC tag on a fresh $(m^*, i^*)$ not produced by $\emulsion.\broadcast$. In $G_2$, $\mathcal{A}$ must use the honest $\cert_{K_0}$ (certificate forgeries are excluded by $G_1$). The TESLA security condition (Definition~\ref{def:chain-security}) ensures $\mathcal{A}$ cannot access $K_{i^*-d}$ before its disclosure window expires, so any valid MAC forgery breaks PRF security of $F'$. By Lemma~\ref{lem:mac}: $|\Pr[G_1=1] - \Pr[G_2=1]|
\leq \ell \cdot \mathsf{Adv}^{\mathsf{PRF}}_{F',\mathcal{B}_2}(\lambda)$.  

\noindent\textit{\underline{$G_2 \to G_3$: Abort on chain forgery}.} $G_3$ additionally aborts if $\mathcal{A}$ submits a fraudulent $K^*_{\mathsf{disc}}$ passing the chain check $F^j(K^*_{\mathsf{disc}}) \stackrel{?}{=} K_0$ in $\emulsion.\auth$. Any such submission breaks one-wayness or TCR of $F$. By Lemma~\ref{lem:chain}: $|\Pr[G_2=1] - \Pr[G_3=1]| \leq \mathsf{Adv}^{\mathsf{OW\text{-}TCR}}_{F, \mathcal{B}_3}(\lambda)$.

\noindent\textit{\underline{$G_3$: No winning strategy}.} Certificate forgeries, MAC forgeries, and chain key forgeries are all excluded in $G_3$, covering every verification path through $\emulsion.\auth$ (Algorithm~\ref{alg:emulsion-verify}). Therefore
$\Pr[G_3=1] \leq \mathsf{negl}(\lambda)$.

\noindent\textit{\underline{Composition}.} Applying the triangle inequality: $\mathsf{Adv}^{\mathsf{EUF\text{-}CMA}}_{\emulsion,
\mathcal{A}} = \Pr[G_0=1] \leq \Pr[G_1=1] + \mathsf{Adv}^{\mathsf{EUF\text{-}CMA}}_{\sgn, \mathcal{B}_1} \leq \Pr[G_2=1] + \ell \cdot \mathsf{Adv}^{\mathsf{PRF}}_{F', \mathcal{B}_2} + \mathsf{Adv}^{\mathsf{EUF\text{-}CMA}}_{\sgn,
\mathcal{B}_1} \leq \mathsf{negl}(\lambda) + \mathsf{Adv}^{\mathsf{OW\text{-}TCR}}_{F, \mathcal{B}_3} + \ell \cdot \mathsf{Adv}^{\mathsf{PRF}}_{F', \mathcal{B}_2} + \mathsf{Adv}^{\mathsf{EUF\text{-}CMA}}_{\sgn, \mathcal{B}_1}$, which yields the stated bound. $\qedhere$
\end{proof}


\section{POWDER Setup}\label{app:powder}
We collected the mobility dataset using the POWDER wireless testbed by deploying an experiment consisting of a core network and multiple base stations installed on rooftop locations, located at the University of Utah~\cite{powder}. A separate experiment instance was configured with a campus shuttle bus carrying an srsUE device. The shuttle follows a fixed route at regular intervals, enabling controlled and repeatable mobility scenarios across the coverage areas of different base stations. As the shuttle moves away from one base station and approaches another with a stronger signal quality, a handover procedure is triggered, allowing the UE to transition to the better signal provider. During the experiments, all network communications and signaling exchanges were collected as packet capture (PCAP) traces. These traces were subsequently parsed using TShark to generate a structured dataset containing detailed mobility information, including the transmission and reception timestamps of SIB1 messages between the base station and UE, identification of handover events, and related signaling procedures. Additionally, the average speed of the campus shuttle was logged at approximately 29 mph, while the UE-to-base-station distance was estimated using the signal propagation time derived from the collected traces.

\section{MIB Coverage via Deferred Verification.}\label{app:mib}
Although the MIB is carried on the PBCH rather than the DL-SCH, 
its payload is only 3\,bytes (24\,bits)~\cite{RRCSpec}, small 
enough to be included directly in the per-packet HMAC input 
without affecting packet size. Concretely, the BS computes 
$\tau_i \leftarrow \mathsf{HMAC}(K'_i,\; m \| m_{\mathsf{MIB}} 
\| i \| \mathit{ID}_{\mathit{BS}})$, where $m_{\mathsf{MIB}}$ 
denotes the current MIB content. The UE necessarily decodes the 
MIB \emph{before} receiving SIB1 (it is required to locate 
SIB1 on the DL-SCH), so MIB protection is inherently 
retroactive: the UE buffers the decoded MIB and, upon 
successful HMAC verification of the corresponding SIB1 packet, 
checks that the HMAC verifies under the MIB it decoded 
from the PBCH. Since $m_{\mathsf{MIB}}$ is an input to 
the HMAC, successful verification confirms that the BS 
committed to the same MIB content the UE received; a 
verification failure indicates a spoofed MIB and causes 
the UE to abort the connection attempt, preventing any 
further protocol progression (RACH, RRC, NAS) on a 
fraudulent cell. This deferred-verification model mirrors 
TESLA's own design, where packets are buffered pending key 
disclosure, and adds zero communication overhead, since the 
3-byte MIB content is absorbed into the existing HMAC 
computation at negligible cost.

\section{Relay and Downgrade Attacks}\label{app:relaydowngrade}
\smallskip\noindent\textbf{Relay Attacks.} A relay forwarding genuine, unmodified SIBs may pass our verification, since Algorithm~\ref{alg:emulsion-verify} binds the tag to $\mathit{ID}_{\mathit{BS}}$ but not to physical-layer observables (PCI, ARFCN, RF characteristics).  However, because the replayed content is authentic, the adversary cannot manipulate cell-selection parameters, reselection priorities, or public-warning payloads without invalidating the HMAC tag; only attacks that require mere UE presence on the rogue cell survive.  Binding authentication to a physical-layer anchor is complementary and left to future work.

\smallskip\noindent\textbf{Downgrade Attacks.} During transition or partial rollout, an FBS can broadcast a legacy unauthenticated SIB1, indistinguishable from a non-upgraded cell. We can propose two alternatives here. (i) A UE in strict mode rejects unauthenticated SIBs, preventing downgrade at the cost of denying service on non-upgraded cells; (ii) in permissive mode it prefers authenticated cells but falls back to unauthenticated ones, preserving coverage but allowing bypass. EMULSION supports both via a per-PLMN policy flag provisioned through the eSIM channel: strict enforcement applies only to PLMNs the operator has enrolled, so an FBS cannot downgrade enrolled networks by omitting authentication material, while unenrolled networks remain accessible.

\end{document}